\documentclass[11pt,reqno]{article}
\usepackage[utf8]{inputenc}
\usepackage{booktabs} % for much better looking tables
\usepackage{array} % for better arrays (eg matrices) in maths
\usepackage{paralist} % very flexible & customisable lists (eg. enumerate/itemize, etc.)
\usepackage{verbatim} % adds environment for commenting out blocks of text & for better verbatim
\usepackage{amsthm}
\usepackage[table,xcdraw]{xcolor}
\usepackage[titletoc,toc,title]{appendix}
\usepackage{latexsym, amsmath, amssymb, a4, epsfig, color}
\usepackage{blindtext}
\usepackage{graphicx}
\usepackage{subcaption}
\usepackage[utf8]{inputenc}
\usepackage[export]{adjustbox}
\usepackage{wrapfig}
\usepackage{apptools}
\usepackage{afterpage}
\usepackage{tabularx}
\usepackage{floatrow}
\usepackage{bm}
\usepackage{amssymb}
\usepackage{amsmath}
\usepackage[normalem]{ulem} % to use strikethrough %\sout{text}
\usepackage{algorithm}
\usepackage[noend]{algpseudocode}
\usepackage[T1]{fontenc}
\usepackage[font=small,labelfont=bf,tableposition=top]{caption}

\DeclareCaptionLabelFormat{andtable}{#1~#2  \&  \tablename~\thetable}
\AtAppendix{\counterwithin{definition}{subsection}}
\AtAppendix{\counterwithin{theorem}{subsection}}
\newcolumntype{C}[1]{>{\centering\arraybackslash}p{#1}}
\newfloatcommand{capbtabbox}{table}[][\FBwidth]

\graphicspath{}

\newtheorem{theorem}{\bf Theorem}[section]

\newtheorem{cor}{\bf Corollary}[section]

\newtheorem{lemma}{\bf Lemma}[section]

\newcommand{\RR}{\mathbb{R}}
\newcommand{\CC}{\mathbb{C}}
\newcommand{\NN}{\mathbb{N}}
\newcommand{\ZZ}{\mathbb{Z}}

\newcommand{\Om}{\Omega}
\newcommand{\ds}{\displaystyle}

\newcommand{\p}{\partial}
\newcommand{\pd}[2]{\frac {\p #1}{\p #2}}
\newcommand{\eqnref}[1]{(\ref {#1})}
\renewcommand{\qed}{\hfill $\Box$ \medskip}
\newcommand{\beq}{\begin{equation}}
\newcommand{\eeq}{\end{equation}}

\newcommand{\SingleOmega}{\mathcal{S}_{\partial\Omega}}

\newcommand{\KstarOmega}{\mathcal{K}_{\partial\Omega}^{*}}
\newcommand{\KOmega}{\mathcal{K}_{\partial\Omega}}
\newcommand{\Kcal}{\mathcal{K}}
\newcommand{\Scal}{\mathcal{S}}
\newcommand{\Dcal}{\mathcal{D}}

\newcommand{\Rtwo}{\mathbb{R}^2}

\newcommand{\la}{\langle}
\newcommand{\ra}{\rangle}

\newcommand{\RRe}{\operatorname{Re}}
\newcommand{\IIm}{\operatorname{Im}}

\newcommand{\iu}{\mathrm{i}}

\numberwithin{equation}{section}
\numberwithin{figure}{section}

\begin{document}

\newcommand{\TheTitle}{An extension of the Eshelby conjecture to domains of general shape in anti-plane elasticity}
\newcommand{\TheAuthors}{Doosung Choi, Kyoungsun Kim and Mikyoung Lim}
\title{{\TheTitle}\thanks{{
The third author is the corresponding author. This work was supported by the National Research Foundation of Korea(NRF) funded by the Ministry of Education (NRF-2016R1A2B4014530 and NRF-2019R1A6A1A10073887).
}}}
\author{
Doosung Choi\thanks{\footnotesize Department of Mathematical Sciences, Korea Advanced Institute of Science and Technology, Daejeon 34141, Korea ({7john@kaist.ac.kr},  {mklim@kaist.ac.kr})}\and 
Kyoungsun Kim\thanks{\footnotesize Department of Mathematical Sciences, Seoul National University, Seoul 08826, Korea (kgsunsis@snu.ac.kr)} \and Mikyoung Lim\footnotemark[2]}
\date{\today}
\maketitle

\begin{abstract}
According to the Eshelby conjecture, an ellipse or ellipsoid is the only shape that induces an interior uniform strain under a uniform far-field loading. We extend the Eshelby conjecture to domains of general shape for anti-plane elasticity. Specifically, we show that for each positive integer $N$, an inclusion induces an interior uniform strain under a polynomial loading of degree $N$ if and only if the exterior conformal map of the inclusion is a Laurent series of degree $N$. Furthermore, for the isotropic case, we characterize the shape of an inclusion by only using the first-degree polynomial loading and explicitly solve the interior potential of the inclusion in terms of the Grunsky coefficients.
%Some numerical experiments are presented to illustrate our results.
\end{abstract}
%%%%%%%%%%%%%%%%%%%%%%%%%%%
%\subject{applied mathematics, materials science, mathematical physics}


%\tableofcontents

\section{Introduction}\label{sec:intro}
We consider the field perturbation due to the presence of an inclusion in a homogeneous background. An inclusion whose material parameter is different from that of the background brings a field perturbation to its exterior and interior. 
 The resulting perturbation depends on the shape of the inclusion as well as the material parameter, and certain shapes admit extremal properties. In 1957, Eshelby discovered that an ellipsoid embedded in an infinite elastic medium induces a uniform interior strain for uniform loadings \cite{Eshelby:1957:DEF}. Then, he made the following conjecture, which is known as {\it the Eshelby conjecture}, in \cite{Eshelby:1961:EII}: ``Among closed surfaces, the ellipsoid alone has this convenient property.''
In this study, we extend this uniformity property to inclusions of general shape by using higher-order loadings, for anti-plane elasticity.

The Eshelby conjecture was proven by Sendeckyj \cite{Sendeckyj:1970:EIP} for two dimensions and by Ru and Schiavone \cite{Ru:1996:EIA} for anti-plane elasticity.
To prove the conjecture, Ru and Schiavone \cite{Ru:1996:EIA} and Sendeckyj \cite{Sendeckyj:1970:EIP} used complex analytic function theory. Kang and Milton \cite{Kang:2008:SPS} provided an alternative proof using the hodograph transformation.
Various non-ellipsoidal shapes were shown not to satisfy the Eshelby uniformity property in three dimensions. For example, Rodin \cite{Rodin:1996:EIP} considered polyhedral inclusions, Markenscoff \cite{Markenscoff:1997:SEI} inclusions with a planar piece on their boundary, and Lubarda and Markenscoff \cite{Lubarda:1998:AEP} non-convex inclusions. Markenscoff obtained that the space of domains satisfying the Eshelby uniformity property forms a nine-dimensional manifold \cite{Markenscoff:1998:ICE}.
Finally, the conjecture for three dimensions was proven by Kang and Milton \cite{Kang:2008:SPS} and Liu \cite{Liu:2008:SEC} in relation to the Newtonian potential. We recommend that readers refer to the review article by Kang \cite{Kang:2009:CPS} to determine more relations of the Eshelby conjecture with the P\'{o}lya-Szeg\"{o} conjecture and the Newtonian potential problem.

Shapes other than ellipses or ellipsoids can also satisfy the uniformity property with a modified condition from the Eshelby conjecture. Finding Eshelby inclusions, which denote inclusions undergoing a uniform eigenstress with either a far-field loading or a modified condition, has practical applications for designing composites that result in small variances in internal stresses.
A multiply connected inclusion can satisfy the uniformity property \cite{Cherepanov:1974:IPP, Kang:2008:IPS, Lee:1977:ESE, Liu:2008:SEC}, and so can a non-elliptical simply connected inclusion on a bounded domain containing the inclusion with some boundary condition \cite{Bardsley:2017:CGB, Kang:2011:SBV, Lim:2018:NNC}. Kang \emph{et al.} constructed Eshelby inclusions with two disjoint components in two dimensions using the hodograph transformation technique \cite{Kang:2008:IPS}, and Liu designed multiply connected ones in two and three dimensions with a variational approach \cite{Liu:2008:SEC}.
Ru derived analytic solutions for Eshelby inclusions of arbitrary shape in a plane or half-plane in terms of some complex analytic functions \cite{Ru:1999:ASE,Ru:2000:EPT}. Kang \emph{et al}. \cite{Kang:2011:SBV} and Bardsley \emph{et al}. \cite{Bardsley:2017:CGB} designed various Eshelby inclusions embedded in a bounded domain using the hodograph transformation technique.
Wang \emph{et al}. found Eshelby inclusions of arbitrary shape with the traction-free condition on a curvilinear boundary \cite{Wang:2018:EIA}. Lim and Milton found Eshelby inclusions of arbitrary shape embedded in a bounded domain using the conformal mapping technique \cite{Lim:2018:NNC}; see also \cite{Milton:2001:NCI, Wang:2019:ECE}. We refer to the works by Vigdergauz \cite{Vigdergauz:1976:IEI} and by Grabovsky and Kohn \cite{Grabovsky:1995:MMEb} for Vigdergauz microstructures, which are inclusions with the uniformity property with periodic boundary conditions.
Furthermore, analytic and numerical methods to compute the elastic tensor (often called the Eshelby tensor field) for inclusions of various shapes have been developed \cite{Chen:2015:NEE, Gao:2010:SGS, Huang:2011:EEE, Huang:2009:EEE,  Lee:2016:EEF, Onaka:2001:AET, Zou:2010:EPN}. 
Additionally, note that the uniformity property can play a significant role in imaging problems. For example, a location search algorithm for a ball-shaped anomaly was developed by using the fact that the induced electric field is uniform inside a ball \cite{Kwon:2002:RTA}.

In this paper, we investigate the shape of an inclusion which undergoes a uniform eigenstress with a far-field polynomial loading of given degree. If an inclusion has a general shape such that the exterior conformal mapping corresponding to this inclusion is a Laurent series with terms of order $\leq-2$, then the resulting solution due to a uniform loading contains some terms of degree $\geq2$ in the interior of the inclusion (see, for example, \cite{Ammari:2019:SRN}).
Instead, the stress field is uniform for a polynomial loading as shown in \cite{Sendeckyj:1970:EIP} for the plane elastostatic problem; one can observe this case from \cite{Ru:1996:EIA} for anti-plane elasticity.
In fact, the Eshelby conjecture can be generalized to characterize the shape of which an inclusion undergoes a uniform eigenstress with a polynomial loading of given degree (for details, see Theorem~\ref{theorem:main:a}).
To the best of our knowledge, there have been no reports on extending the Eshelby conjecture to provide a characterization scheme for inclusions of general shape. Furthermore, we explicitly find the polynomial loading which induces a uniform strain inside the inclusion in a simple form by using the Faber polynomials. For the isotropic case, we also explicitly express the field perturbation using the Grunsky coefficients and find a characterization scheme for the shape of the inclusion by using only a first-degree polynomial loading.

Our analysis is based on the geometric series expansions of the layer potential operators, recently developed by Jung and Lim \cite{Jung:2018:SSM}. This method provides a new powerful scheme to address the conductivity inclusion problems. We emphasize that with the density basis functions constructed in \cite{Jung:2018:SSM}, the interface problem of anti-plane elasticity in the presence of a simply connected inclusion can be reformulated to a matrix problem. This matrix formulation gives us explicit relations between the exterior conformal mapping of the inclusion and the applied far-field loading, given that the resulting field is uniform inside the inclusion.
It is worth remarking that by using this solution method, one can construct neutral inclusions of multi-layer structure \cite{Choi:2018:GME:preprint} and derive an asymptotic formula to approximate the shape of an inclusion by considering the inclusion as a perturbation of its equivalent ellipse \cite{Choi:2020:ASR}. The decay property of eigenvalues of the Neumann-Poincar\'{e} operator was also obtained \cite{Jung:2020:DEE}.

The remainder of this paper is organized as follows. We state the main results in section \ref{sec:main}. 
Section \ref{sec:Pre} provides the boundary integral formulation for the conductivity transmission problem and the geometric series expansions for the layer potential operators. 
In section \ref{subsection:mainproof}, we derive essential relations for the density function in the boundary integral formulation. 
The proofs of the main results are presented in section \ref{sec:proof}. We finish with the conclusion in section \ref{ref:conclusion}.

\section{Main results}\label{sec:main}
Let $\Om$ be a simply connected, bounded planar domain with $C^{1,\alpha}$ boundary for some $\alpha\in(0,1)$.
We assume $\Om$ has a constant, possibly anisotropic, conductivity $\bm{\sigma}$. 
We consider the transmission problem
\beq\label{cond_eqn0}
\begin{cases}
\ds\nabla\cdot\left(\bm{\sigma}\chi(\Om)+\bm{I}\chi(\RR^2\setminus\overline{\Om})\right)\nabla u=0\quad&\mbox{in }\RR^2\setminus \partial \Omega, \\
\ds u(x) - H(x)  =O({|x|^{-1}})\quad&\mbox{as } |x| \to \infty
\end{cases}
\end{equation}
for a given far-field loading $H$, which is an entire harmonic function. The symbol $\chi$ indicates the characteristic function, and $\bm{I}$ is the $2\times2$ identity matrix.
We assume $\bm{I}-\bm{\sigma}$ is either positive or negative definite and, hence, \eqnref{cond_eqn0} is solvable (see \cite{Escauriaza:1993:RPS}).
One can easily show that $u$ satisfies
\beq\label{cond_eqn0:bc}
u\big|^+=u\big|^-,\quad\nu\cdot\nabla u\big|^+=\nu\cdot\bm{\sigma}\nabla u\big|^-\quad\mbox{on }\p\Om,
\eeq
where $\nu$ denotes the outward unit normal vector to $\p\Om$, and the symbols $+$ and $-$ indicate the limit from the exterior and interior of $\Om$, respectively.

Let us introduce some terminology before stating main results. 
We identify $x=(x_1,x_2)$ in $\RR^2$ with $z=x_1+\iu x_2$ in $\CC$. The symbols $\operatorname{Re}$ and $\operatorname{Im}$ indicate the real and imaginary parts of complex numbers, respectively.
From the Riemann mapping theorem, there exists a unique $\gamma>0$ and a conformal mapping $\Psi$ from $\{w\in\CC:|w|>\gamma\}$ onto $\CC\setminus\overline{\Om}$ satisfying $\Psi(\infty)=\infty$ and $\Psi'(\infty)=1$. 
 This map admits the Laurent series expansion
\beq\label{conformal:Psi}
\Psi(w)=w+a_0+\frac{a_1}{w}+\frac{a_2}{w^2}+\cdots
\eeq
with complex coefficients $a_k$; one can find the derivation in \cite[Chapter 1.2]{Pommerenke:1992:BBC:book}.
The exterior conformal mapping $\Psi$ determines the so-called Faber polynomials $F_m(z)$, which are monic polynomials of degree $m$ determined by $a_0,\dots,a_{m-1}$ and form a basis for analytic functions in $\Om$ \cite{Smirnov:1968:FCV}. 
We define a formal infinite series associated with $\Om$ in terms of Faber polynomials as
$$\mathfrak{F}(z):=\sum_{m=2}^\infty \frac{\overline{a_m}}{\gamma^{2m}} F_m(z).$$

For an ellipse, $a_m=0$ for all $m\geq2$ and the corresponding formal infinite series $\mathfrak{F}(z)$ is the zero function. By generalizing this characterization of an ellipse, we define classes of shape:
%\begin{definition}
For each $N\in\NN$, we call $\Om$ a {\it domain of order $\it N$} if
\beq\label{Psi:finite:N}
\Psi(w)=w+a_0+\frac{a_1}{w}+\cdots+\frac{a_N}{w^N}\ (a_N\neq0).
\eeq
We call a disk, as well as ellipses, a domain of order $1$. 
%\end{definition}
If $\Om$ is a domain of order $N\geq2$, then $\mathfrak{F}(z)$ is a polynomial of degree $N$.
On the other hand, if the far-field loading is a real harmonic polynomial of degree $N$, then it can be expressed as
\beq \label{H:expan}
H(x)=\frac{1}{2}\sum_{m=0}^{N} \left(\alpha_m F_m(z)+\overline{\alpha_m F_m(z)}\right)
\eeq
with some complex coefficients $\alpha_m$. Here, $\alpha_m=0$ for all $m\ge N+1$.

 The main object of this study is to find an equivalent condition for $\Om$ to induce a uniform strain inside $\Om$ under a far-field loading of given finite degree.
First, we extend the Eshelby conjecture in anti-plane elasticity and characterize domains of higher-order (see Theorem \ref{theorem:main:a} below). One can consider this result as an extension of {\it the strong Eshelby conjecture}, following the terminology in \cite{Kang:2008:SPS}, in the sense that the uniformity property for just one loading implies the shape of the inclusion. The proofs of the main results are provided in section~\ref{sec:proof}.
\begin{theorem}[Anisotropic case]\label{theorem:main:a}
Assume that $\Om$ is a simply connected, bounded planar domain with $C^{1,\alpha}$ boundary for some $\alpha\in(0,1)$ and $\bm{\sigma}$ is possibly anisotropic. 
For arbitrary $N\in\NN$, the followings are equivalent:
\begin{itemize}
\item[\rm(a)]
$\Om$ is a domain of order $N$.
\item[\rm(b)]
For some polynomial loading $H$ of degree $N$, the solution $u$ to \eqnref{cond_eqn0} has a uniform strain inside $\Om$.
\end{itemize}
\end{theorem}

Furthermore, we can explicitly find the far-field loading functions that induce the uniformity in terms of Faber polynomials. To state the formula, we need the following $2\times2$ real matrix
\begin{align}
\boldsymbol{\tau}(t)=
(1-2t)
\begin{bmatrix}
\ds\operatorname{Re}\left\{{\tau_1}(t)\right\} &-\operatorname{Im}\left\{{\tau_2}(t)\right\}\\[2mm]
\ds\operatorname{Im}\left\{{\tau_1}(t)\right\} &~~\operatorname{Re}\left\{{\tau_2}(t)\right\}
\end{bmatrix}
\end{align}
with
%where $\tau_1(t)$ and $\tau_2(t)$ are complex functions given by
\begin{align}
{\tau_1}(t)=\frac{\dfrac{a_1}{\gamma^2}+2t}{\left|\dfrac{a_1}{\gamma^2}\right|^2-4t^2}\label{def:beta1},\quad
{\tau_2}(t)=\frac{-\dfrac{a_1}{\gamma^2}+2t}{\left|\dfrac{a_1}{\gamma^2}\right|^2-4t^2}.
\end{align}
As is popularly known as the Bieberbach conjecture \cite{Bieberbach:1916:KDP}, the coefficient $a_1$ of the exterior conformal mapping $\Psi$ satisfies
\beq\notag
|a_1|<\gamma^2.
\eeq
Hence, $\boldsymbol{\tau}(t)$ is invertible for all $t\in(-\infty,-1/2]\cup(1/2,\infty)$.
 In particular, $\bm{\tau}(-1/2)$ is invertible.

\begin{theorem}[Anisotropic case]\label{theorem:main:bc}
Let $\Om$ be a domain of arbitrary finite order and $\bm{\sigma}$ be possibly anisotropic.
For any real constant vector $(e_1,e_2)$, $u$ admits the uniform interior strain $\nabla u=(e_1,e_2)$ for the polynomial loading $H(x)$ given by
\begin{align}\label{H:aniso:cond}
H(x)=f_1x_1+f_2x_2+c_1\operatorname{Re}\big\{z+{{\tau_1}(-1/2)}\, \mathfrak{F}(z)\big\}+c_2\operatorname{Im}\big\{z-{{\tau_2}(-1/2)}\, \mathfrak{F}(z)\big\}
\end{align}
with
\begin{equation}\label{f:condition}
\begin{gathered}
\begin{bmatrix}
\ds f_1\\
\ds f_2
\end{bmatrix}
={\bm\sigma}
\begin{bmatrix}
\ds e_1\\
\ds e_2
\end{bmatrix},\quad
\begin{bmatrix}
\ds c_1\\
\ds c_2
\end{bmatrix}
=\boldsymbol{\tau}(-1/2)^{-1}
(\bm{I}-\bm{\sigma})
\begin{bmatrix}
\ds e_1\\
\ds e_2
\end{bmatrix}.
\end{gathered}
\eeq
In fact, $H$ given by \eqnref{H:aniso:cond} are the only far-field loadings of finite degree that induce a uniform  strain in $\Om$.
\end{theorem}
Interestingly, for a domain of any finite order, the far-field loading that admits a uniform eigenstress is a linear combination of two functions.
We also emphasize that the formulations \eqnref{H:aniso:cond}--\eqnref{f:condition} depend only on $\mathfrak{F}(z)$ and $\bm{\sigma}$.

If $\bm{\sigma}$ is isotropic, we can also characterize the shape of an inclusion by only using the first-degree polynomial loadings as follows. 
%%%%
\begin{theorem}[Isotropic case]\label{theorem:iso:polnomial:a}
Assume $\bm{\sigma}=\sigma\bm{I}$ and set $\lambda=\frac{\sigma+1}{2(\sigma-1)}$. Let $u$ be the solution to \eqnref{cond_eqn0}.
For arbitrary $N\in\NN$, the followings are equivalent:
\begin{itemize}
\item[\rm(a)]
$\Om$ is a domain of order $N$.
\item[\rm(b)]
For some polynomial loading $H$ of degree $1$, the function $(\lambda+\frac{1}{2})(u-H)-\Dcal_{\p\Om}\left[(u-H)|_{\p\Om}\right]$ is a polynomial of degree $N$ in $\Om$. 
\end{itemize}
Here, $\Dcal_{\p\Om}$ denotes the double-layer potential associated with $\Om$ (see \eqnref{def:doublelayer}). 
\end{theorem}

In addition, we can explicitly express $u$ as a Faber polynomial expansion for a given polynomial loading of arbitrary degree. The expansion coefficients have a matrix form as follows.
\begin{theorem}[Isotropic case]\label{theorem:iso:polnomial:b}
Let $\Om$ be a $C^{1,\alpha}$ domain for some $\alpha\in(0,1)$. Assume $\bm{\sigma}=\sigma\bm{I}$ and set $\lambda=\frac{\sigma+1}{2(\sigma-1)}$. 
For a polynomial loading $H$ given by \eqnref{H:expan} with arbitrary degree $N$, the solution $u$ to \eqnref{cond_eqn0} can be expanded as
$$u(x)=\frac{1}{2}\sum_{m=0}^{\infty}\left(\beta_m F_m(z)+\overline{\beta_m F_m(z)}\right)\quad\mbox{in }{\Om}$$
with 
$$\bm{\beta} = \left( \lambda-\frac{1}{2} \right) \left( \lambda\bm{\alpha}\gamma^{2\NN} + \frac{1}{2} \overline{\bm{\alpha} C} \right)  \left( \lambda^2 I - \frac{\gamma^{-2\NN}C\gamma^{-2\NN}\overline{C}}{4} \right)^{-1} \gamma^{-2\NN},$$
where $\bm{\alpha} = (\alpha_m)_{m=1}^\infty$ and $\bm{\beta} = (\beta_m)_{m=1}^\infty$ are infinite row vectors, and $C=(c_{mk})_{m,k=1}^\infty$ is a semi-infinite matrix defined by the Grunsky coefficients $c_{mk}$ of $\Om$ (see \eqnref{eqn:Faberdefinition} for the definition), and $\gamma^{-2\NN}$ is a semi-infinite diagonal matrix whose $(m,m)$-entry is $\gamma^{-2m}$ for each $m\in\NN$.
\end{theorem}
%%%%%%%%%%%%%%%%%%

Figures~\ref{fig:aniso}--\ref{fig:iso2} illustrate the result in Theorem \ref{theorem:main:bc}. In all examples, the potential difference in $u$ between the neighboring level curves is $1/2$. 
The loading function $H$ is given by \eqnref{H:aniso:cond}--\eqnref{f:condition} such that $(f_1+c_1,f_2+c_2)=(0,1)$ in Figure~\ref{fig:aniso}; $(c_1,c_2)=(1,0)$ in the first row of Figure~\ref{fig:iso2}; $(c_1,c_2)=(0,1)$ in the second row of Figure~\ref{fig:iso2}.

\captionsetup[subfigure]{labelformat=empty}
\begin{figure}[H]
\centering
\begin{subfigure}{0.35\textwidth}
\centering
\includegraphics[height=3.5cm]{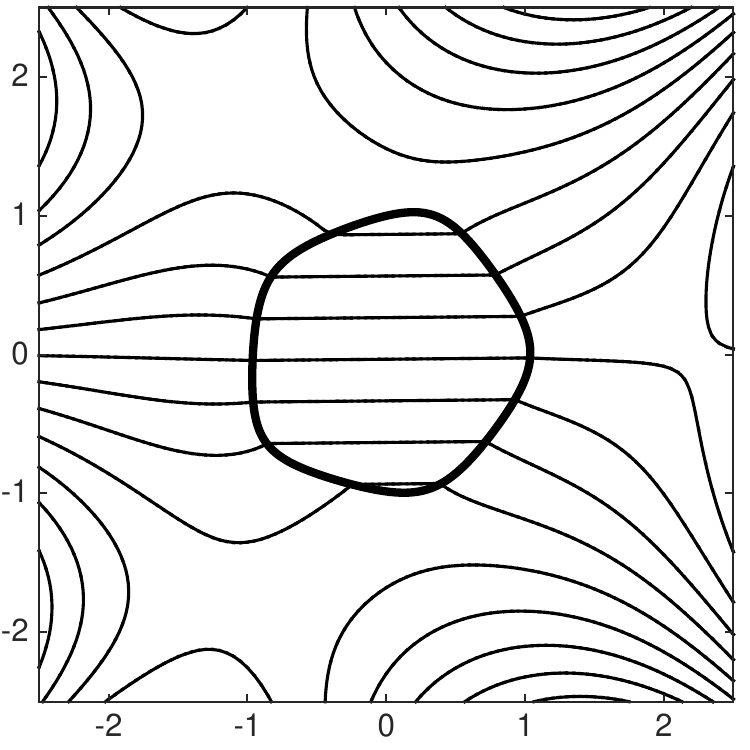}
\caption{$\sigma=\begin{bmatrix}0.99  &0 \\ 0& 0.2\end{bmatrix}$}
\end{subfigure}
\begin{subfigure}{0.35\textwidth}
\centering
\includegraphics[height=3.5cm]{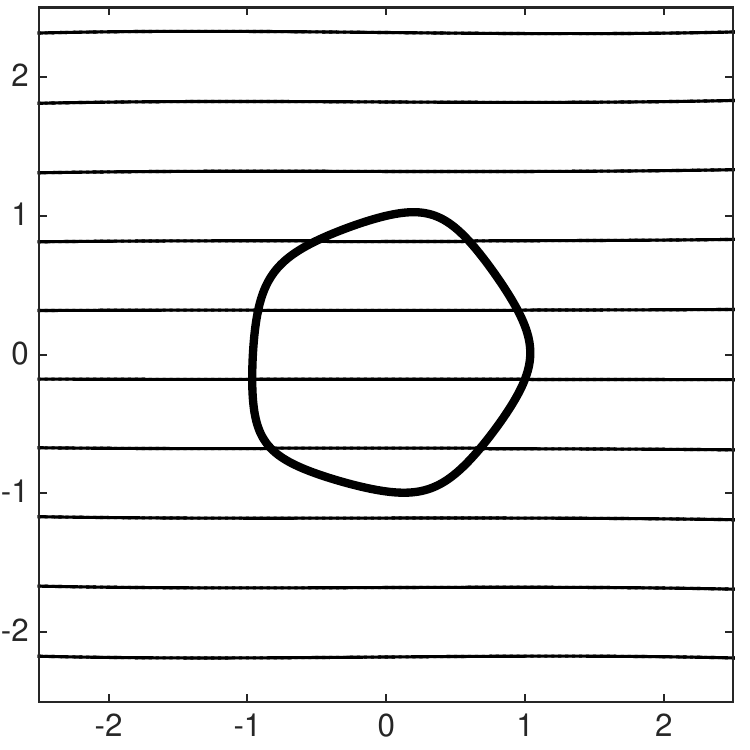}
 \caption{$\sigma=\begin{bmatrix}0.2  &0 \\ 0& 0.99\end{bmatrix}$}
\end{subfigure}
\caption{Anisotropic case. Thin curves are level curves of $u$ and thick curves indicate $\partial \Om$. 
The domain $\Om$ is of order $N=4$, where $\gamma=1$, $a_1=0.1\iu*s$, $a_2=-0.05\iu*s$, $a_3=0$, and $a_4=0.15*s$ with $s=\frac{1}{4}$. 
 }
\label{fig:aniso}
%elseif casenumber == 4
%    gamma =1; M=50;
%    a_dim =4;
%    mu0 = 0; mu = zeros(1,3*M);
%
%    mu_vec(1)=  0.1*1i;
%    mu_vec(2) =-0.05i;
%    mu_vec(3)=0;
%    mu_vec(4) =0.15;
%    scalefactor = 4;
%
%    for j = 1:4
%        mu(j)=mu_vec(j)/scalefactor;
%    end
%    if subcasenumber ==1
%        sigma1 = [0.2 0; 0 0.2];
%    elseif subcasenumber ==2
%        sigma1 = [0.2  0.1 ; 0.1  0.2];
%    elseif subcasenumber ==3
%        sigma1 = [0.99  0 ; 0 0.2];
%    elseif subcasenumber == 4
%        sigma1 = [0.2  0; 0 0.99];
%
%    end
%    exterior_direction = [0;1];
%exterior_direction = [1;0]; % node_number = 1000
\end{figure}

\begin{figure}[H]
\begin{subfigure}{\textwidth}
\centering
\captionsetup{labelformat=empty}
\begin{minipage}{0.22\linewidth}
\includegraphics[width=\linewidth]{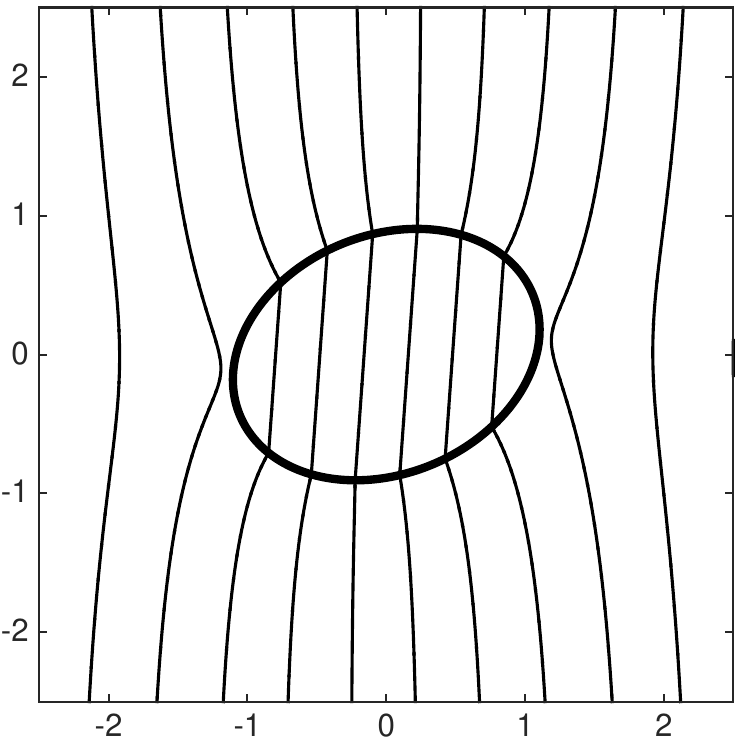}\caption{$N=1$}
\end{minipage}\qquad
\begin{minipage}{0.22\linewidth}
\includegraphics[width=\linewidth]{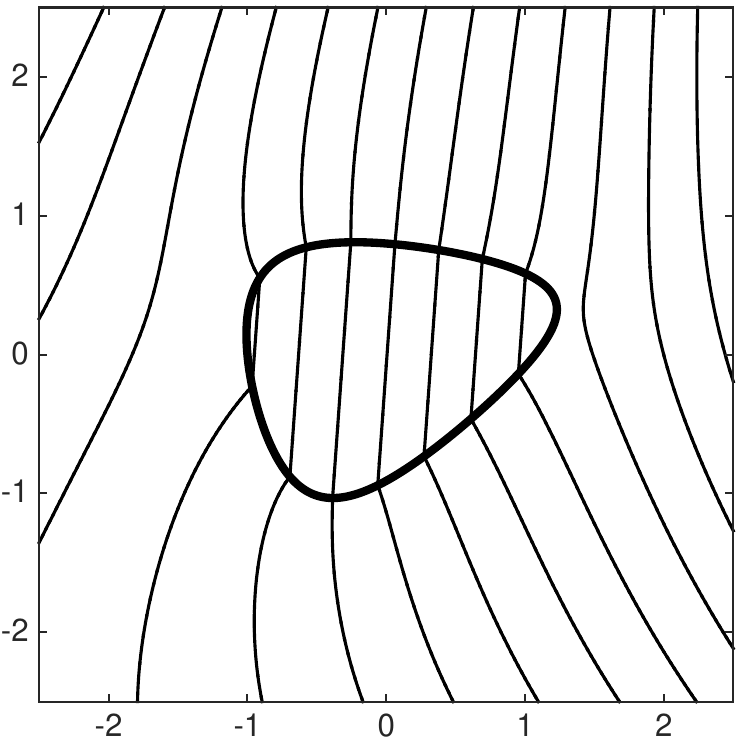}\caption{$N=2$}
\end{minipage}\qquad
\begin{minipage}{0.22\linewidth}
\includegraphics[width=\linewidth]{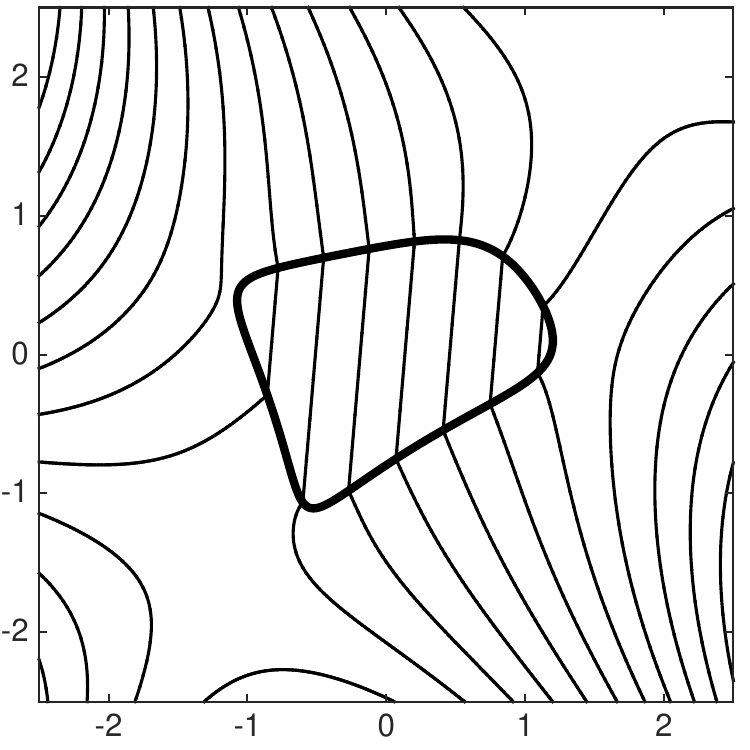}\caption{$N=3$}
\end{minipage}
%\begin{minipage}{0.205\linewidth}
%\includegraphics[width=\linewidth]{Eshelby_iso14.pdf}\caption{$N=4$}
%\end{minipage}\qquad
\end{subfigure}\vskip 3mm
%\caption{Isotropic case. Level curves of $u$ are drawn, where $\Om$ is a domain of order $N=1,2,3$. 
%The coefficients, $\{a_j\}_{j\leq N}$, are given by $a_1=0.1+0.1\iu$, $a_2=0.1+0.1\iu$, $a_3=-0.1\iu$, and $a_4= 0.05$. For all examples, $\sigma=0.2$, $\gamma=1$,
%$H$ is defined by \eqnref{H:aniso:cond}--\eqnref{f:condition} with $(c_1,c_2)=(1,0)$.
%}\label{fig:iso1}
%\end{figure}
%\begin{figure}[b!]
\begin{subfigure}{\textwidth}
\centering
\captionsetup{labelformat=empty}
\begin{minipage}{0.22\linewidth}
\includegraphics[width=\linewidth]{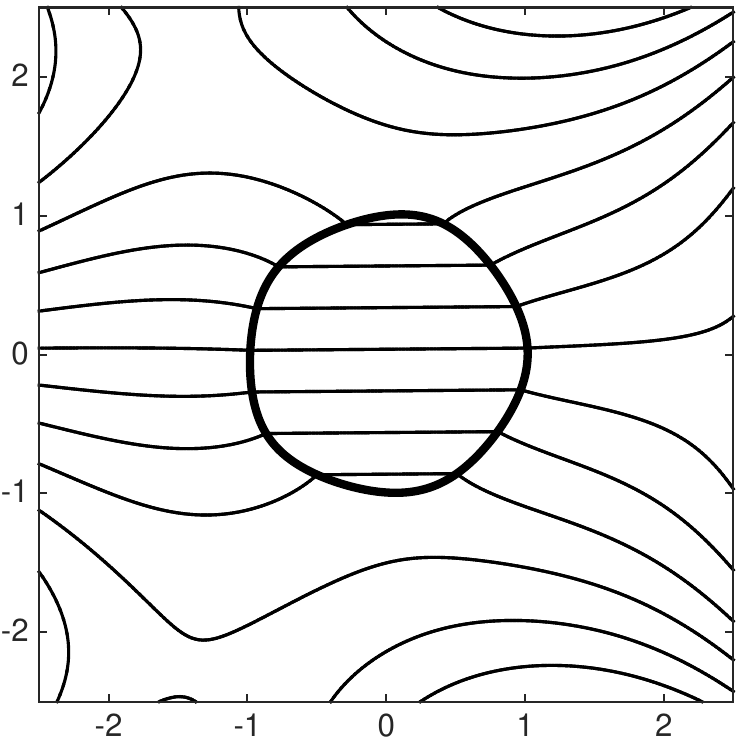}\caption{$s=\frac{1}{8}$}
\end{minipage}\qquad
\begin{minipage}{0.22\linewidth}
\includegraphics[width=\linewidth]{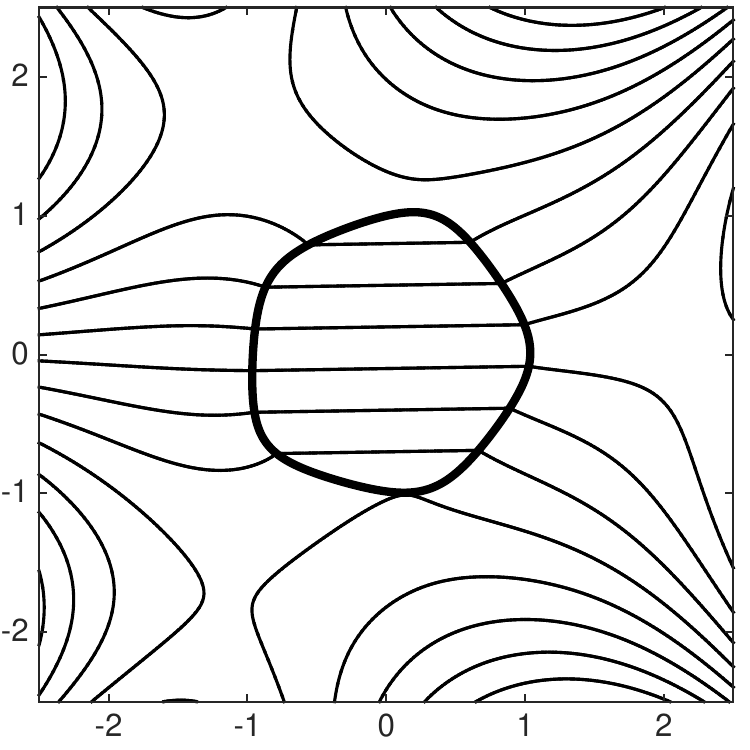}\caption{$s=\frac{1}{4}$}
\end{minipage}\qquad
\begin{minipage}{0.22\linewidth}
\includegraphics[width=\linewidth]{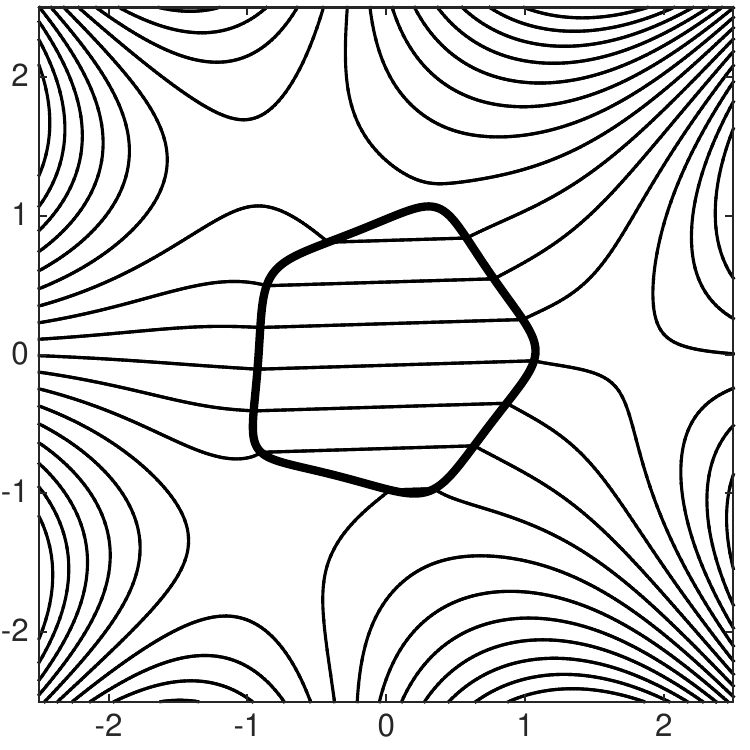}\caption{$s=\frac{1}{2}$}
\end{minipage}
%\begin{minipage}{0.205\linewidth}
%\includegraphics[width=\linewidth]{Eshelby_iso21.pdf}\caption{$s=1$}
%\end{minipage}\qquad
\end{subfigure}
\caption{Isotropic case ($\sigma=0.2$). 
In the first row, $\Om$ is a domain of order $N=1,2,3$ given by \eqnref{Psi:finite:N}, where $\gamma=1$, $a_1=0.1+0.1\iu$, $a_2=0.1+0.1\iu$, and $a_3=-0.1\iu$. % and $a_4= 0.05$. 
In the second row, $\Om$ is given as in Figure~\ref{fig:aniso} with variant $s$; as $\Om$ more resembles a disk with smaller $s$, the corresponding $H$ has smaller higher-order coefficients.
}
\label{fig:iso2}
\end{figure}

\section{Preliminary}\label{sec:Pre}

\subsection{Boundary integral formulation for the conductivity transmission problem}\label{sec:Boundary}

Let $\Omega$ be a Lipschitz domain in $\RR^2$. 
For $\varphi\in L^2(\p \Omega)$, the single- and double-layer potentials are defined as
\begin{align}\notag
\Scal_{\p \Omega}[\varphi](x)&=\int_{\partial \Omega}\Gamma(x-y)\varphi(y)\,d\sigma(y),\quad x\in\Rtwo,\\[1mm]
\Dcal_{\p \Omega}[\varphi](x)&=\int_{\partial \Omega}\frac{\partial}{\partial\nu_y}\Gamma(x-y)\varphi(y)\,d\sigma(y),\quad x\in\Rtwo\setminus\partial \Omega,\label{def:doublelayer}
\end{align}
where $\Gamma$ is the fundamental solution to the Laplacian, {i.e.}, $\Gamma(x)=\frac{1}{2\pi}\ln|x|$, and $\nu_y$ denotes the outward unit normal vector to $\partial \Omega$. 
The single-layer potential has the jump relation
\begin{align}\label{eqn:jump}
	 \ds\frac{\partial}{\partial\nu}\Scal_{\p \Omega}[\varphi]\Big|^{\pm}&=\left(\pm\frac{1}{2}I+\Kcal^*_{\p \Omega}\right)[\varphi]\quad\mbox{on }\p \Omega,
	 \end{align}
where $I$ denotes the identity operator on $L^2(\p \Omega)$ and $\Kcal_{\p \Omega}^*$ is the so-called Neumann--Poincar\'{e} (NP) operator. The NP operator is the boundary integral operator defined as
\begin{align*}%\label{def:Kstar}
	 \ds\Kcal^*_{\p \Omega}[\varphi](x)&=\text{p.v.}\,\frac{1}{2\pi}\int_{\partial \Omega}\frac{\left<x-y,\nu_x\right>}{|x-y|^2}\varphi(y)\,d\sigma(y),
%		\ds	\KOmega[\varphi](x)&=\text{p.v.}\,\frac{1}{2\pi}\int_{\partial\Omega}\frac{\left<y-x,\nu_y\right>}{|x-y|^2}\varphi(y)\,d\sigma(y).
\end{align*}
where p.v. stands for the Cauchy principal value.

We set $\Scal_{\p \Omega}[\varphi](z):=\Scal_{\p \Omega}[\varphi](x)$ for $z=x_1+\iu x_2$. Likewise, we define $\Dcal_{\p \Omega}[\varphi](z)$ and $\Kcal^*_{\p \Omega}[\varphi](z)$.

\smallskip

\noindent\textbf{Isotropic case.}
When $\bm{\sigma}=\sigma\bm{I}$, $0<\sigma\neq1<\infty$, the solution $u$ to \eqnref{cond_eqn0} can be expressed as
\beq\label{u:integral}
u(x)=H(x)+\Scal_{\p\Om}[\varphi](x)\quad\mbox{in }\RR^2,
\eeq
where $\varphi\in L^2_0(\p\Om)$ (i.e., $\varphi$ is square-integrable and has a zero average value) is given by
\beq\label{varphi:eqn}
\varphi=(\lambda I-\Kcal_{\p\Om}^*)^{-1}\left[\nu\cdot \nabla H\right]\quad\mbox{with }\lambda = \frac{\sigma+1}{2(\sigma-1)}.
\eeq
The operator $\Kcal^*_{\p \Omega}$ is bounded on $L^2(\p \Omega)$, and $\lambda I-
\Kcal_{\p\Om}^*$ is invertible on $L^2_0(\p \Om)$ for $|\lambda|\geq1/2$ \cite{Escauriaza:1992:RTW,Kellogg:1953:FPT,Verchota:1984:LPR}.
For more properties of the NP operator, we refer the readers to \cite{Ammari:2013:MSM:book, Ammari:2004:RSI:book, Kenig:1994:HAT} and the references therein.

\smallskip

\noindent\textbf{Anisotropic case.}
We now assume that $\bm{\sigma}$ is anisotropic. In other words, $\bm{\sigma}$ is a $2\times2$ positive definite symmetric matrix satisfying that $\bm{I}-\bm{\sigma}$ is either positive or negative definite.
For $\psi\in L^2(\p\Om)$, we define the single-layer potential associated with $\bm{\sigma}$ as 
$$\Scal^{\bm{\sigma}}_{\p\Om}[\psi](x)=\int_{\p\Om} \Gamma^{\bm{\sigma}}(x-y) \psi(y) \,d\sigma(y)\quad\mbox{in }\RR^2,$$
where 
$$ \Gamma^{\bm{\sigma}}(x)=\frac{1}{2\pi \sqrt{\mbox{det} (\bm{\sigma})}} \ln\left|\bm{\sigma}^{-\frac{1}{2}}x \right|$$
 and $\bm{\sigma}^{-\frac{1}{2}}$ is the inverse of the square root matrix of $\bm{\sigma}$.
The solution $u$ to \eqnref{cond_eqn0} with the anisotropic conductivity $\bm{\sigma}$ can be expressed as (see \cite{Escauriaza:1993:RPS})
\beq\label{sol:aniso}
u(x) =
\begin{cases}
\ds H(x)+\Scal_{\p \Om}[\varphi](x)\quad&\mbox{in }\RR^2\setminus\overline{\Om},\\
\ds \Scal^{\bm{\sigma}}_{\p\Om}[\psi](x)\quad&\mbox{in }\Om,
\end{cases}
\end{equation}
where the density functions $(\psi,\varphi) \in L^2(\p \Om) \times L^2_0(\p \Om)$ satisfy
\begin{equation*}
\begin{cases}
\ds\Scal^{\bm{\sigma}}_{\p\Om}[\psi] - \Scal_{\p \Om}[\varphi]= H, \\
\ds\nu \cdot \bm{\sigma} \nabla \Scal^{\bm{\sigma}}_{\p\Om}[\psi]\Big|^{-}- \nu \cdot \nabla \Scal_{\p\Om} [\varphi]\Big|^{+} = \nu \cdot \nabla H\quad\mbox{on }\p\Om.
\end{cases}
\end{equation*}

\subsection{Geometric series expansions of the layer potential operators}\label{sec:series_solution}

The exterior conformal mapping $\Psi$ associated with $\Om$ (see \eqnref{conformal:Psi}) uniquely defines the so-called {\it Faber polynomials} $F_m(z)$ via the generating relation
\beq\label{faber:generating}
\frac{w\Psi'(w)}{\Psi(w)-z}=\sum_{m=0}^\infty \frac{F_m(z)}{w^{m}}\quad \mbox{for }z\in{\Om},\ |w|>\gamma.
\eeq
Then, $F_m(\Psi(w))$ has only one positive term $w^m$, i.e.,
\begin{equation}\label{eqn:Faberdefinition}
	F_m(\Psi(w))%=w^m+\frac{c_{m,1}}{w}+\frac{c_{m,2}}{w^2}+\cdots=
	=w^m+\sum_{k=1}^{\infty}c_{mk}{w^{-k}},
\end{equation}
where $c_{mk}$ are called the {\it Grunsky coefficients}. 
The symmetric relation $$kc_{mk}=mc_{km}$$
 holds for all $m,k\in\NN$.
The concept of Faber polynomials was first introduced by G. Faber \cite{Faber:1903:PE} and has been one of the essential elements in geometric function theory (see  \cite{Duren:1983:UF}).

Each $F_m$ is an $m$-th order monic polynomial uniquely determined by the coefficients $a_0,a_1,\cdots,a_{m-1}$ of $\Psi$; a recursive formula to compute the coefficients of Faber polynomials is well known \cite[Ch.4]{Duren:1983:UF}. The first three Faber polynomials are
%$F_0(z)=1$, $F_1(z)=z-a_0$, and $F_2(z)=z^2-2a_0 z+(a_0^2-2a_1)$.
\beq\notag
F_0(z)=1,\quad F_1(z)=z-a_0,\quad F_2(z)=z^2-2a_0 z+(a_0^2-2a_1).
\eeq
From the fact that $F_1(z)=z-a_0$ and the symmetricity, we have 
$c_{1k}=a_k$ and
\beq\label{C_k1:a}
c_{k1}=kc_{1k}=ka_k\quad\mbox{for all }k\in\NN.
\eeq
Because $F_m(z)$ is a monic polynomial of order $m$ (the highest order term is $z^m$), we obtain the following lemma by substituting \eqnref{Psi:finite:N} into \eqnref{eqn:Faberdefinition}.
\begin{lemma}\label{lemma:finite}
Let $\Om$ be a domain of order $N\geq2$. Then,
 $c_{m,Nm}=(a_N)^m\neq 0$
 and $c_{mk}=c_{km}=0$ for all $m\in\NN$, $k\geq Nm+1$.
\end{lemma}

We recall that $\Om$ is assumed to be $C^{1,\alpha}$. The continuous extension of the conformal mapping to the boundary is well known \cite{Caratheodory:1913:GBR}. 
Furthermore, by the Kellogg-Warschawski theorem \cite{Pommerenke:1992:BBC:book}, $\Psi'$ also has continuous extension to the boundary.
We define a curvilinear orthogonal coordinates $(\rho,\theta)\in[\rho_0,\infty)\times[0,2\pi)$ via the relation
\beq\label{coord:Psi}
z=\Psi(e^{\rho+\iu\theta})\quad\mbox{for } z\in\CC\setminus\Om.
\eeq
For simplicity, we set $v(\rho,\theta)=v\left(\Psi(e^{\rho+\iu\theta})\right)$ for a complex function $v$.
  The scale factors with respect to $\rho$ and $\theta$ coincide with each other. We denote them by
	\begin{equation}\notag
	h(\rho, \theta) := \left|\frac{\partial \Psi} {\partial \rho}\right| = \left|\frac{\partial \Psi} {\partial \theta}\right|.
	\end{equation}	
On $\p\Om$, it then holds that
$$d\sigma(z)=h(\rho_0,\theta)d\theta$$ and
 \beq
 	\frac{\partial v}{\partial \nu}\Big|_{\p\Om}^{+}(z)=\frac{1}{h(\rho, \theta)}\frac{\partial }{\partial \rho}v(\rho,\theta)\Big|_{\rho\rightarrow\rho_0^+}.\label{eqn:normalderiv}
 \eeq
We denote $\la\cdot,\cdot\ra$ as the inner product in $L^2(\p\Om,h)$, the weighted $L^2$ space with the weight function $h$. In other words, for functions $p,q$ on $\p\Om$ satisfying $\int_{\p\Om}|p|^2 hd\sigma,\;  \int_{\p\Om}|q|^2 hd\sigma <\infty$, we define
\begin{align}\label{eqn:inner_product}
\la p,q\ra&=\frac{1}{2\pi}\int_{\p\Om}p(z)\overline{q(z)}h(z)d\sigma(z).
%&=\frac{1}{2\pi}\int_0^{2\pi}p(\rho_0,\theta)\overline{q(\rho_0,\theta)} (h(\rho_0,\theta))^2d\theta.
\end{align}
For each $m\in\ZZ$, we define the density function
$$\psi_m(z):=\frac{e^{\iu m\theta}}{h(\rho_0,\theta)}\quad\mbox{on }\p\Om.$$
They are orthogonal with respect to the inner product \eqnref{eqn:inner_product}, i.e., 
$\la \psi_m,\psi_n\ra=\delta_{mn}$. 

We can express the layer potential operators of the density function $\psi_{m}$ in terms of the Faber polynomials and the Grunsky coefficients as follows.

\begin{lemma}[\cite{Jung:2018:SSM}]\label{thm:series}
Let $\Om$ be a simply connected, bounded planar domain with $C^{1,\alpha}$ boundary for some $\alpha\in(0,1)$.
We identify $z=\Psi(w)\in\CC\setminus{\Om}$ with $(\rho,\theta)$ via the relation \eqnref{coord:Psi}.
\begin{itemize}
\item[\rm(a)]
We have
\beq\label{Scal_zeta0}
\Scal_{\p\Om}[\psi_0](z)=
\begin{cases}
\ln \gamma \quad &\mbox{for }z\in\overline{\Om},\\
\ln|w|\quad&\mbox{for }z\in\CC\setminus\overline{\Om}.
\end{cases}
\eeq

For $m=1,2,\dots$, we have
\begin{align}\label{eqn:seriesSLpositive}
\SingleOmega[\psi_m](z)&=
\begin{cases}
\ds-\frac{1}{2m\gamma^m}F_m(z)\quad&\text{for }z\in\overline{\Omega},\\[2mm]
\ds-\frac{1}{2m\gamma^m}\left(\sum_{k=1}^{\infty}c_{mk}e^{-k(\rho+\iu\theta)}+\gamma^{2m}e^{m(-\rho+\iu\theta)}\right)\quad &\text{for } z\in\CC\setminus\overline{\Om},
\end{cases}\\[3mm]
\label{eqn:seriesSLnegative}
\SingleOmega[\psi_{-m}](z)&=
\begin{cases}
\ds	-\frac{1}{2m\gamma^{m}}\overline{F_{m}(z)}\quad&\text{for }z\in\overline{\Omega},\\[2mm]
\ds-\frac{1}{2m\gamma^m}\left(\sum_{k=1}^{\infty}\overline{c_{mk}}e^{-k(\rho-\iu\theta)}+\gamma^{2m}e^{m(-\rho-\iu\theta)}\right)\quad &\text{for } z\in\CC\setminus\overline{\Om}.		
\end{cases}
\end{align}
The two series converge uniformly on $\{(\rho,\theta):\rho\geq\rho_1\}$ for any fixed $\rho_1>\ln\gamma$.

\item[\rm(b)]
We also have
$$\Kcal^*_{\p\Om}[\psi_0]=\dfrac{1}{2}\psi_0, \quad  \KstarOmega[{\psi_m}]=\frac{1}{2}\sum_{k=1}^{\infty}\frac{k}{m}\frac{c_{mk}}{\gamma^{m+k}}\, {\psi}_{-k}, \quad  \KstarOmega[{\psi_{-m}}]=\frac{1}{2}\sum_{k=1}^{\infty}\frac{k}{m}\frac{\overline{c_{mk}}}{\gamma^{m+k}}\, {\psi}_{k},
$$
where the infinite series converges in the Sobolev space $H^{-1/2}(\p\Om)$.
\end{itemize}
\end{lemma}

The following lemma is essential for characterizing a domain of finite order.
\begin{lemma}\label{psi_1_poly}
For any $N\geq2$, $\Om$ is a domain of order $N$ if and only if $\KstarOmega[\psi_1]\in L^2(\p\Om,h)$ and 
\begin{align*}
\left\la \KstarOmega[\psi_1],\psi_{- N}\right\ra \neq 0, \quad \left\la \KstarOmega[\psi_1],\psi_{- k}\right\ra =0\quad\mbox{for all }k>N.
\end{align*}
Moreover, $\Om$ is a domain of order $1$ (i.e., a disk or ellipse) if and only if $\KstarOmega[\psi_1]\in L^2(\p\Om,h)$ and
$$
\left\la \KstarOmega[\psi_1],\psi_{- k}\right\ra =0\quad\mbox{for all }k>1.
$$
\end{lemma}
\begin{proof}
Note that $\Om$ is a domain of finite order $N$ if and only if $a_k=0$ for all $k\geq N+1$
and $a_N\neq 0$ (except $N=1$). 
Recall that $c_{1k}=a_k$.
From the orthogonality $\la \psi_m,\psi_n\ra=\delta_{mn}$, we have
\begin{align}\label{NP_series1}
\left\la \KstarOmega[\psi_1],\psi_{-k}\right\ra=\frac{k}{2}\frac{c_{1k}}{\gamma^{1+k}}=\frac{k}{2}\frac{a_k}{\gamma^{1+k}}.
\end{align}
 Hence, we complete the proof.
\end{proof}

%%%%%%%%%%%%%%%%%%%%%%%%%

\section{Density relations}\label{subsection:mainproof}
In this section, we assume that $H$ and $u|_{\overline{\Om}}$ are harmonic polynomials of degree $L$ and $K$, respectively. For the anisotropic case, only $K=0, 1$ is considered.

As shown in \eqnref{H:expan}, it holds that for some complex coefficients $\alpha_m$ and $\beta_m$,
\begin{align*}
H(x)
&=\frac{1}{2}\sum_{m=0}^{L}\left(\alpha_m F_m(z)+\overline{\alpha_m F_m(z)}\right)\quad\mbox{in }\CC,\\
u(x)&=\frac{1}{2}\sum_{m=0}^{K}\left(\beta_m F_m(z)+\overline{\beta_m F_m(z)}\right)\quad\mbox{on }\overline{\Om}.
\end{align*}
We note from \eqnref{eqn:seriesSLpositive} and \eqnref{eqn:seriesSLnegative} that
\begin{align*}
&\Scal_{\p\Om}\left[-m\gamma^m \psi_m\right](z)=\frac{1}{2}F_m(z),\\
&\Scal_{\p\Om}\left[-m\gamma^m \psi_{-m}\right](z)=\frac{1}{2}\overline{F_m(z)} \quad\mbox{on }\overline{\Om}. 
\end{align*}
Hence, $H$ and $u$ satisfy
\begin{align}
H(x)\label{eqn:H:single}
&=\frac{1}{2}(\alpha_0+\overline{\alpha_0})+\Scal_{\p\Om}[-\psi_H](x),\\
u(x)\label{eqn:u:single}
&=\frac{1}{2}(\beta_0+\overline{\beta_0})+\Scal_{\p\Om}[-\psi_u](x)\quad\mbox{on }\overline{\Om}
\end{align}
with 
\begin{align}\label{psi:H:expan}
{\psi_H}&=\sum_{m=1}^{L} \left(\alpha_m m \gamma^m\psi_m+\overline{\alpha_m}m\gamma^m\psi_{-m}\right),\\
\label{psi:u:expan}
\psi_u&=\sum_{m=1}^{K} \left(\beta_m m \gamma^m \psi_m+\overline{\beta_m} m \gamma^m\psi_{-m}\right).
\end{align}
%In the following, we derive relations between $\psi_H$ and $\psi_u$. 
It then follows from the jump formula of the single-layer potential \eqnref{eqn:jump} that
\begin{align}\label{H:deri:psi}
&\pd{H}{\nu}=\left(\frac{1}{2}I-\Kcal^*_{\p\Om}\right)[\psi_H],\\
&\pd{u}{\nu}\Big|^-=\left(\frac{1}{2}I-\Kcal^*_{\p\Om}\right)[\psi_u]\quad\mbox{on }\p\Om.
\end{align}

The following relation is useful in deriving the relations between $\psi_H$ and $\psi_u$:
\beq\label{nu_complex}
\frac{1}{2}\left(\nu_1+\iu\nu_2\right)=\left(\frac{1}{2}I-\Kcal^*_{\p\Om}\right)[\gamma\psi_1].
\eeq
Indeed, from \eqnref{eqn:seriesSLpositive} with $m=1$, it holds that
$\Scal_{\p\Om}[\psi_1](z)=-\frac{1}{2\gamma}F_1(z)=-\frac{1}{2\gamma}(x_1+\iu x_2 -a_0).$
By taking the interior normal derivative, we have
{
$$
\left(-\frac{1}{2}I+\Kcal^*_{\p\Om}\right)[\psi_1]
=\frac{\partial}{\partial\nu}\Scal_{\p \Om}[\psi_1]\Big|^{-}
=-\frac{1}{2\gamma}(1,\iu)\cdot\nu=-\frac{1}{2\gamma}(\nu_1+\iu\nu_2).
$$}
This implies \eqnref{nu_complex}.

\begin{lemma}\label{lemma:main_method}
Assume that $H$ is a harmonic polynomial of degree $L$ and that $u$ is a harmonic polynomial of degree $K$ inside $\Om$ ($K=0,1$ for the anisotropic case). We set $\psi_H$ and $\psi_u$ as in \eqnref{psi:H:expan} and \eqnref{psi:u:expan}, respectively.
We set $\varphi$ to be the density function on $\p\Om$ satisfying \eqnref{u:integral} for the isotropic case or \eqnref{sol:aniso} for the anisotropic case. Then, we have
\beq\label{eqn:varphi:psi:tpsi}
\varphi=\psi_H-\psi_u.
\eeq
Furthermore, the followings hold:
\begin{itemize}
\item[\rm(a)]
For the isotropic case, we have % ${\bm{\sigma}}$ is isotropic, then we have
\beq\label{eqn:main_relation}
\Kcal^*_{\p\Om}[\psi_u]=\lambda\psi_u-\Big(\lambda-\frac{1}{2}\Big)\psi_H.
\eeq

\item[\rm(b)]
For the anisotropic case, assuming $\nabla u=(e_1,e_2)$ in $\Om$ for some real constants $e_1,e_2$, we have  %If ${\bm{\sigma}}$ is anisotropic and $\nabla u$ is constant in $\Om$, then ${\psi_H}$ satisfies
\begin{align}
\Kcal^*_{\p\Om}\left[\, \overline{(e-f)}\, \gamma\psi_1+(e-f)\gamma\psi_{-1}\right]\label{eqn:aniso:relation:ulinear}
=-\frac{1}{2}\, \overline{(e+f)}\, \gamma\psi_1 -\frac{1}{2}(e+f)\gamma\psi_{-1}+\psi_H
\end{align}
with $e=e_1+\iu e_2$, $f=f_1+\iu f_2$ and $(f_1, f_2)=\bm{\sigma} (e_1, e_2)$ in $\Om$. 
\end{itemize}
\end{lemma}
\begin{proof}
Because $\Scal_{\p\Om}[\varphi](x)=u(x)-H(x)$ for $x\in\RR^2\setminus\overline{\Om}$ (even when ${\bm{\sigma}}$ is anisotropic), we have from \eqnref{eqn:H:single} and \eqnref{eqn:u:single} that
$$\Scal_{\p\Om}[\varphi]=\Scal_{\p\Om}[-{\psi_u}]-\Scal_{\p\Om}[-\psi_H]+\mbox{const.}\quad\mbox{in }\Om.$$
Indeed, the equality holds for $x\in\p\Om$. Since both sides are harmonic in $\Om$ and continuous on $\overline{\Om}$, the equality also holds in $\Om$ from the uniqueness in the boundary value problem of harmonic functions.
By taking the interior normal derivative, we obtain
$$\left(-\frac{1}{2}I+\Kcal^*_{\p\Om}\right)[\varphi]=\left(-\frac{1}{2}I+\Kcal^*_{\p\Om}\right)[\psi_H-{\psi_u}]\quad\mbox{on }\p\Om.$$
Because $-\frac{1}{2}I+\Kcal^*_{\p\Om}$ is invertible on $L^2_0(\p\Om)$ and $\varphi, \psi_H-{\psi_u}$ are in $L^2_0(\p\Om)$, we deduce \eqnref{eqn:varphi:psi:tpsi}.

\smallskip

If $\bm{\sigma}$ is isotropic, the relations \eqnref{varphi:eqn} and \eqnref{H:deri:psi} imply that
$$(\lambda I-\Kcal^*_{\p\Om})[\varphi]=\pd{H}{\nu}=\left(\frac{1}{2}I-\Kcal^*_{\p\Om}\right)[\psi_H].$$
From \eqnref{eqn:varphi:psi:tpsi}, it is straightforward to obtain \eqnref{eqn:main_relation}.

\smallskip

Now, let ${\bm{\sigma}}$ be anisotropic
and $\nabla u=(e_1,e_2)$ in $\Om$. Then, we have
$$u(x_1,x_2)= \frac{1}{2}\overline{e}z+\frac{1}{2} e \overline{z}  + \mbox{constant}$$
and
\beq\label{eqn:tpsi}
{\psi_u}=\overline{e}\gamma\psi_1+e\gamma\psi_{-1}.
\eeq
The definition of $(f_1,f_2)$ and \eqnref{nu_complex} imply that
\begin{align}\notag
\nu\cdot{\bm{\sigma}}\nabla u\big|^-
&=\nu_1f_1+\nu_2f_2\\\notag
&=\frac{1}{2}(\nu_1+\iu\nu_2)\overline{f}+\frac{1}{2}(\nu_1-\iu\nu_2)f\\\label{eqn:nusigmagradu}
&=\left(\frac{1}{2}I-\Kcal^*_{\p\Om}\right)\left[\,  \overline{f}\gamma\psi_1+f\gamma\psi_{-1}\right].
\end{align}
Meanwhile, we have from \eqnref{sol:aniso}, \eqnref{H:deri:psi}, \eqnref{eqn:varphi:psi:tpsi} and the transmission condition of $u$ on $\p\Om$ that
\begin{align*}\nu\cdot{\bm{\sigma}}\nabla u\big|^-
=\nu\cdot\nabla u\big|^+
=\pd{H}{\nu}+ \pd{}{\nu}\Scal_{\p\Om}[\varphi]\Big|^+
=-\Kcal^*_{\p\Om}[{\psi_u}]+\psi_H-\frac{1}{2}{\psi_u}.
\end{align*}
Using \eqnref{eqn:tpsi} and \eqnref{eqn:nusigmagradu}, we deduce
\beq\notag%\label{eqn:anistropic:mainrelation1}
\left(\frac{1}{2}I-\Kcal^*_{\p\Om}\right)\left[\overline{f}\gamma\psi_1+f\gamma\psi_{-1}\right]
+\left(\frac{1}{2}I+\Kcal^*_{\p\Om}\right)\left[\overline{e}\gamma\psi_1+e\gamma\psi_{-1}\right]=\psi_H,
\eeq
and this implies \eqnref{eqn:aniso:relation:ulinear}.
\end{proof}

 Ru and Schiavone proved the Eshelby uniformity conjecture for the anti-plane elasticity \cite{Ru:1996:EIA}. 
We provide an alternative proof by using Lemmas~\ref{psi_1_poly} and \ref{lemma:main_method} as follows.
\begin{cor}[The Eshelby Conjecture]\label{cor:eshelby}
%We assume the same regularity condition for $\Om$ as in Lemma~\ref{thm:series}.
For any $\bm{\sigma}$, either isotropic or anisotropic, $\Om$ is an ellipse if and only if the solution $u$ to \eqnref{cond_eqn0} has a uniform strain in $\Om$ for a uniform loading $H$.
\end{cor}
\begin{proof}
We only prove that $\Om$ is an ellipse if $u$ has a uniform strain in $\Om$ for a uniform loading $H$.
From the assumption, we have $$K=0\mbox{ or }1, \quad\mbox{and}\quad L=1,$$
and it follows from Lemma~\ref{lemma:main_method}(b) that
\beq\label{eqn:eshelby:proof:relation}
\Kcal^*_{\p\Om}\left[\, \overline{(e-f)}\, \gamma\psi_1+(e-f)\gamma\psi_{-1}\right]=-\frac{1}{2}\, \overline{(e+f)}\, \gamma\psi_1 -\frac{1}{2}(e+f)\gamma\psi_{-1}+\psi_H,
\eeq
where $e=e_1+\iu e_2$, $f=f_1+\iu f_2$ and $(f_1,f_2)=\bm{\sigma} \nabla u=\bm{\sigma} (e_1,e_2)$ in $\Om$.

We have $e\neq0$ and, i.e., $$K\neq0.$$ Indeed, if $e=0$, then $f=0$ from the definition of $f$. This implies that $\psi_H=0$ from \eqnref{eqn:eshelby:proof:relation}. This contradicts the assumption that $L=1$ (which implies that $\psi_H\neq0$). 

From the invertibility of $\bm{I}-\bm{\sigma}$, we then deduce $e\neq f$.
Note that the right-hand side of \eqnref{eqn:eshelby:proof:relation} belongs to the linear space spanned by $\{\psi_{1},\psi_{-1}\}$. By taking the inner product with $\psi_{-k}$ (see \eqnref{eqn:inner_product}) for both sides of \eqnref{eqn:eshelby:proof:relation}, we observe that
$$\overline{(e-f)}\, \gamma\left\la \KstarOmega[\psi_1],\psi_{-k}\right\ra=0\quad\mbox{for all }k\geq2. $$
Thanks to Lemma~\ref{psi_1_poly}, we prove that $\Om$ is an ellipse.
\end{proof}

\section{Proof of the main results}\label{sec:proof}

\subsection{Proof of Theorem~\ref{theorem:main:a} and Theorem~\ref{theorem:main:bc}} \label{sec:proof:Thm21}
We start with the partial proof of Theorem~\ref{theorem:main:a} for the isotropic case. 
\begin{lemma}[Isotropic case] Let $\bm{\sigma}=\sigma\bm{I}$ with $0<\sigma \neq 1<\infty$, and set $\lambda=\frac{\sigma+1}{2(\sigma-1)}$. For the statements {\rm(a)} and {\rm(b)} in Theorem~\ref{theorem:main:a}, we have {\rm(b)} implies {\rm (a)}.
\end{lemma}
\begin{proof} We assume (b): $H$ is a harmonic polynomial of degree $N$ and $u$ has a uniform strain inside $\Om$, i.e., 
 $$K=0\mbox{ or }1, \quad\mbox{and}\quad L=N,$$
 following the terminology in section \ref{subsection:mainproof}.
Then, we have $K\neq0$ by the same analysis as in the proof of Corollary \ref{cor:eshelby}.  From Lemma \ref{lemma:main_method}\,(a), it then follows that
\begin{align}\notag
&\Kcal^*_{\p\Om}\left[\beta_1\gamma\psi_1 +\overline{\beta_1}\gamma\psi_{-1}\right]\\\notag%\label{eqn:proof:equation1_1}
&=\lambda\left(\beta_1\gamma\psi_{1}+\overline{\beta_1}\gamma\psi_{-1}\right)
-\Big(\lambda-\frac{1}{2}\Big)\sum_{m=1}^N\left(\alpha_m m \gamma^m\psi_m+\overline{\alpha_m} m \gamma^m\psi_{-m}\right)
\end{align}
with some complex coefficients $\beta_1\neq0$ and $\alpha_N\neq0$.
Taking the inner product for both sides with $\psi_{-m}$ and applying Lemma~\ref{thm:series}(b), we arrive that
\beq\label{iso:proof:eqn1}
\beta_1\gamma\left\la \Kcal^*_{\p\Om}[\psi_1],\psi_{-m}\right\ra
=
\begin{cases}
\ds\Big[\lambda\overline{\beta_1}-\Big(\lambda -\frac{1}{2}\Big)\overline{\alpha_1}\Big]\gamma\quad &\mbox{for }m=1,\\
\ds-\Big(\lambda-\frac{1}{2}\Big)\overline{\alpha_m}m\gamma^m\quad&\mbox{for } m=2,\cdots,N,\\
\ds0\quad &\mbox{for }m\geq N+1.
\end{cases}
\eeq
Meanwhile, it also holds from \eqnref{C_k1:a} and \eqnref{NP_series1} that
\beq\label{iso:proof:eqn2}
\beta_1\gamma\left\la \Kcal^*_{\p\Om}[\psi_1],\psi_{-m}\right\ra
=\beta_1\gamma \frac{m}{2}\frac{c_{1m}}{\gamma^{m+1}}=\beta_1 \frac{m}{2} \frac{a_m}{\gamma^m}\quad\mbox{for all } m\in\NN.
\eeq
By comparing \eqnref{iso:proof:eqn1} and \eqnref{iso:proof:eqn2}, we have
$$a_N\neq0,\quad a_m=0\quad \mbox{for all }m\geq N+1.$$
In other words, $\Om$ is a domain of order $N$. Hence, we prove the lemma.

Moreover, equations \eqref{iso:proof:eqn1} and \eqref{iso:proof:eqn2} also imply the algebraic relations:
\begin{align*}
&\frac{a_1}{\gamma^2}\beta_1-2\lambda\overline{\beta_1}=(1-2\lambda)\overline{\alpha_1},\\
&
\frac{a_m}{\gamma^{2m}}\beta_1=(1-2\lambda)\overline{\alpha_m},\quad m=2,\cdots,N,
\end{align*}
which are equivalent to
\begin{align}
\beta_1 &= \frac{(1-2\lambda)}{\left|\dfrac{a_1}{\gamma^2}\right|^2-4\lambda^2}
\left[\frac{\overline{a_1}}{\gamma^2}\overline{\alpha_1}+2\lambda\alpha_1\right]\notag\\
&= \frac{(1-2\lambda)}{\left|\dfrac{a_1}{\gamma^2}\right|^2-4\lambda^2}
\left[\left(\frac{\overline{a_1}}{\gamma^2}+2\lambda\right)\RRe\{\alpha_1\}
+\iu\left(-\frac{\overline{a_1}}{\gamma^2}+2\lambda\right)\IIm\{\alpha_1\}\right]\notag\\\label{eqn:talpha1}
&=(1-2\lambda)\left[\overline{{\tau_1}(\lambda)}\RRe\{\alpha_1\}+\iu\overline{{\tau_2}(\lambda)}\IIm\{\alpha_1\}\right]
\end{align}
and
\begin{align}\notag
\alpha_m &=\frac{1}{1-2\lambda}\frac{\overline{a_m}}{\gamma^{2m}}\overline{\beta_1}\\\label{eqn:talpham}
&=\Big[{{\tau_1}(\lambda)}\RRe\{\alpha_1\}-\iu{{\tau_2}(\lambda)}\IIm\{\alpha_1\}\Big]\frac{\overline{a_m}}{\gamma^{2m}},\quad m=2,\cdots,N,
\end{align}
where ${\tau_1}(\lambda)$ and ${\tau_2}(\lambda)$ are given by \eqnref{def:beta1}.

Set $\nabla u=(e_1,e_2)$ in $\Om$. 
As $u(x)=\frac{1}{2}\beta_1 z+\frac{1}{2}\overline{\beta_1}\overline{z}+\mbox{const.},$
we have $e=e_1+\iu e_2= \overline{\beta_1},$
which implies from \eqnref{eqn:talpha1} that
$$
\begin{bmatrix}
\ds e_1\\
\ds e_2
\end{bmatrix}
=\boldsymbol{\tau}(\lambda)
\begin{bmatrix}
\ds \RRe\{\alpha_1\}\\
\ds -\IIm\{\alpha_1\}
\end{bmatrix}
.$$
By defining $c_1=\RRe\{\alpha_1\}$ and $c_2=-\IIm\{\alpha_1\}$, i.e., 
$$\begin{bmatrix}
\ds c_1\\
\ds c_2
\end{bmatrix}
=\boldsymbol{\tau}(\lambda)^{-1}
\begin{bmatrix}
\ds e_1\\
\ds e_2
\end{bmatrix},
$$
and using the fact that $a_m=0$ for $m\geq N+1$ and $F_1(z)=z-a_0$,
we obtain
\begin{align}
H(x)\notag
&=\frac{1}{2}(\alpha_0+\overline{\alpha_0})+\Scal_{\p\Om}[-\psi_H](x)
=\frac{1}{2}(\alpha_0+\overline{\alpha_0})+\RRe\left\{\alpha_1 F_1(z)+\sum_{m=2}^N \alpha_m F_m(z)\right\}\\\notag
&=\frac{1}{2}(\alpha_0+\overline{\alpha_0})+\RRe\left\{c_1F_1(z)-\iu c_2F_1(z)
+\Big[{{\tau_1}(\lambda)}c_1+\iu{{\tau_2}(\lambda)}c_2\Big]\mathfrak{F}(z)\right\}\\ \label{iso:H:Fexp}
&=c_1\operatorname{Re}\Big\{z+{{\tau_1}(\lambda)}\mathfrak{F}(z)\Big\}
+c_2\operatorname{Im}\Big\{z-{{\tau_2}(\lambda)}\mathfrak{F}(z)\Big\} +\mbox{const.}
\end{align} 
\end{proof}

We now prove Theorem~\ref{theorem:main:a}, where Theorem~\ref{theorem:main:bc} is also proven in the meantime.

\noindent\textbf{Proof of Theorem~\ref{theorem:main:a}.}
We first prove (a) assuming (b): $H$ is a harmonic polynomial of degree $N$ and $u$ has a uniform strain inside $\Om$, i.e., 
 $$K=0\mbox{ or }1, \quad\mbox{and}\quad L=N,$$
 following the terminology in section \ref{subsection:mainproof}.
As in Lemma \ref{lemma:main_method}\,(b), we set $\nabla u=(e_1,e_2)$ and $(f_1,f_2)=\bm{\sigma} \nabla u=\bm{\sigma} (e_1,e_2)$ in $\Om$. We also set $e=e_1+\iu e_2$ and $f=f_1+\iu f_2$.

By the same analysis as in the proof of Corollary \ref{cor:eshelby}, we have $K\neq0$.
As discussed in the proof of Corollary~\ref{cor:eshelby}, we have $e\neq f$.
From \eqnref{eqn:aniso:relation:ulinear} and Lemma~\ref{thm:series}\,(b), it holds that 
$$\left\la \KstarOmega[\psi_1],\psi_{- k}\right\ra =0\quad\mbox{for all }k\geq N+1$$
and $\Om$ is a domain of order $N$ thanks to Lemma~\ref{psi_1_poly}. Hence, we prove (b) implies (a). 

Furthermore, equation~\eqnref{eqn:aniso:relation:ulinear} can be written as
\beq\notag
\Kcal^*_{\p\Om}\Big[\widetilde{\psi}^{\mbox{aniso}}\Big]=\lambda^{\mbox{aniso}}\widetilde{\psi}^{\mbox{aniso}}-\Big(\lambda^{\mbox{aniso}}-\frac{1}{2}\Big)\psi^{\mbox{aniso}}
\eeq
with \begin{align*}
\lambda^{\mbox{aniso}}:=&-\frac{1}{2},\\
\psi^{\mbox{aniso}}:=&\psi_H-\overline{f}\gamma\psi_1-f\gamma\psi_{-1}\\
=&(\alpha_1-\overline{f})\, \gamma\psi_1+(\overline{\alpha_1}-f)\gamma\psi_{-1}
+\sum_{m=2}^N \left(\alpha_m m \gamma^m\psi_m+\overline{\alpha_m}m \gamma^m\psi_{-m}\right)\\
\widetilde{\psi}^{\mbox{aniso}}:=&\overline{(e-f)}\, \gamma\psi_1+(e-f)\gamma\psi_{-1},
\end{align*}
where $\psi_H$ is given by \eqnref{psi:H:expan}.
We recall the density relation for the isotropic case, \eqnref{eqn:main_relation}. 
We can interpret $\alpha_1-\overline{f}$ and $\overline{e-f}$ as $\alpha_1$ and $\beta_1$ in the isotropic case, respectively.
By following the same computation as in the isotropic case, one arrives at the following relations (which correspond to \eqnref{eqn:talpha1} and \eqnref{eqn:talpham}):
\begin{align}\label{eqn:talpha1:aniso}
\overline{e-f}&=\left(1-2\lambda^{\mbox{aniso}}\right)\left[\, \overline{{\tau_1}\left(\lambda^{\mbox{aniso}}\right)}\operatorname{Re}\left\{\alpha_1-\overline{f}\right\}+\iu\overline{{\tau_2}\left(\lambda^{\mbox{aniso}}\right)}\operatorname{Im}\left\{\alpha_1-\overline{f}\right\}\right],\\\label{eqn:talpham:aniso}
\alpha_m &= \bigg[{{\tau_1}\left(\lambda^{\mbox{aniso}}\right)}\operatorname{Re}\left\{\alpha_1-\overline{f}\right\}-\iu{{\tau_2}\left(\lambda^{\mbox{aniso}}\right)}\operatorname{Im}\left\{\alpha_1-\overline{f}\right\}\bigg]\frac{\overline{a_m}}{\gamma^{2m}},\quad m = 2,\cdots,N.
\end{align}

Set $$\alpha_1-\overline{f}=:c_1-\iu c_2.$$
From \eqnref{eqn:talpha1:aniso}, we have
$$
\begin{bmatrix}
\ds e_1-f_1\\
\ds e_2-f_2
\end{bmatrix}
=\boldsymbol{\tau}(-1/2)
\begin{bmatrix}
\ds \RRe\{\alpha_1-\overline{f}\}\\
\ds -\IIm\{\alpha_1-\overline{f}\}
\end{bmatrix}
=\boldsymbol{\tau}(-1/2)
\begin{bmatrix}
\ds c_1\\
\ds c_2
\end{bmatrix}.
$$
In other words, $(c_1,c_2)$ satisfies \eqnref{f:condition}.
We then have
$$
\Scal_{\p\Om}\left[-\psi^{\mbox{aniso}}\right](z)\\
=\mbox{const.}+c_1\operatorname{Re}\big\{z+{{\tau_1}(-1/2)}\mathfrak{F}(z)\big\}+
c_2\operatorname{Im}\big\{z-{{\tau_2}(-1/2)}\mathfrak{F}(z)\big\}
$$
(which corresponds to \eqnref{iso:H:Fexp}) 
and
\begin{align*}
\ds H(x)&=\frac{1}{2}(\alpha_0+\overline{\alpha_0})+\Scal_{\p\Om}[-\psi_H](z)\\
\ds&=\frac{1}{2}(\alpha_0+\overline{\alpha_0})+\Scal_{\p\Om}\left[-\overline{f}\gamma\psi_1-f\gamma\psi_{-1}\right]+\Scal_{\p\Om}\left[-\psi^{\mbox{aniso}}\right](z)\\
\ds&=\mbox{const.}+f_1 x_1+f_2 x_2+c_1\operatorname{Re}\big\{z+{{\tau_1}(-1/2)}\mathfrak{F}(z)\big\}+
c_2\operatorname{Im}\big\{z-{{\tau_2}(-1/2)}\mathfrak{F}(z)\big\}.
\end{align*}

\smallskip
We now prove (a) implies (b). 
Let $\Om$ be a domain of order $N$.
We will construct $H$ with which the corresponding solution $u$ has a uniform strain in $\Om$.
Choose any $(e_1,e_2)\neq(0,0)$ and set $(f_1,f_2)$ and $(c_1,c_2)$ as in \eqnref{f:condition}. We then set $\alpha_1=\overline{f}+c_1-\iu c_2$ and $\alpha_2,\cdots,\alpha_N$ to satisfy \eqnref{eqn:talpham:aniso}
and define $H$ and $\psi_H$ as \eqnref{eqn:H:single} and \eqnref{psi:H:expan} with zero as a constant term.
Then, we have \eqnref{eqn:talpha1:aniso}. Furthermore, it holds that \eqnref{eqn:aniso:relation:ulinear}, i.e., \begin{align}
\Kcal^*_{\p\Om}\left[\, \overline{(e-f)}\, \gamma\psi_1+(e-f)\gamma\psi_{-1}\right]\label{eqn:aniso:relation:ulinear2}
=-\frac{1}{2}\overline{(e+f)}\, \gamma\psi_1 -\frac{1}{2}(e+f)\gamma\psi_{-1}+\psi_H,
\end{align} and \beq\label{H:relation:inOm}
H(x)=\Scal_{\p\Om}[-\psi_H](x)\quad\mbox{in }\Om.
\eeq
Set
\begin{align*}
&\widetilde{\psi}:=\overline{e}\gamma\psi_1+e\gamma\psi_{-1},\\
& \varphi:=\psi_H-\widetilde{\psi},\\
&\widetilde{u}(x)
:=\begin{cases}
H(x)+\Scal_{\p\Om}[\varphi](x)\quad&\mbox{for }x\in\CC\setminus\Om,\\
\Scal_{\p\Om}[-\widetilde{\psi}](x)
\quad&\mbox{for }x\in\Om.
\end{cases}
\end{align*}
It is straightforward to find from \eqnref{eqn:seriesSLpositive} that
$$\widetilde{u}(x)=e_1x_1+e_2x_2 + \mbox{const.} \quad\mbox{for }x\in\Om$$
and, hence,
$$\bm{\sigma}\nabla\widetilde{u}=\bm{\sigma}(e_1,e_2)=(f_1,f_2)\quad\mbox{in }\Om.$$
One can easily show that $\widetilde{u}$ satisfies the boundary transmission condition \eqnref{cond_eqn0:bc} due to \eqnref{nu_complex}, \eqnref{eqn:aniso:relation:ulinear2} and \eqnref{H:relation:inOm}.
Hence, we complete the proof.
\qed

\subsection{Proof of Theorem~\ref{theorem:iso:polnomial:a}}

We define
$$\KOmega[\varphi](x)=\text{p.v.}\,\frac{1}{2\pi}\int_{\partial\Omega}\frac{\left<y-x,\nu_y\right>}{|x-y|^2}\varphi(y)\,d\sigma(y).$$
The NP operator $\Kcal^*_{\p\Om}$ is the $L^2$ adjoint of $\Kcal_{\p\Om}$. Moreover, we have
$$\Scal_{\p\Om}\Kcal^*_{\p\Om}=\Kcal_{\p\Om}\Scal_{\p\Om},$$
which is known as Plemelj's symmetrization principle.
The double-layer potential satisfies the jump relation
$$\Dcal_{\p\Om}[\varphi]\Big|^{\pm}=\left(\mp \frac{1}{2}I+\Kcal_{\p\Om}\right)[\varphi]\quad\mbox{on }\p\Om.$$

Let $H$ be an arbitrary first-order polynomial; then, it holds that
\begin{align*}
H(x)&=\Scal_{\p\Om}\left[-\gamma\alpha\psi_1-\gamma\overline{\alpha}\psi_{-1}\right]+\mbox{const.}\quad\mbox{in }\Om
\end{align*}
for some constant $\alpha\neq0$ and
\beq\label{theorem:general:Eshelby:eqn1}
\nu\cdot\nabla H=\left(\frac{1}{2}I-\Kcal^*_{\p\Om}\right)\left[\gamma\alpha\psi_1+\gamma\overline{\alpha}\psi_{-1}\right].
\eeq
For $\varphi$ defined by \eqnref{varphi:eqn}, we have the relation
\begin{align*}
\Scal_{\p\Om}[\nu\cdot\nabla H]
&=
\Scal_{\p\Om}(\lambda I-\Kcal^*_{\p\Om})[\varphi]\notag\\
&=\lambda \Scal_{\p\Om}[\varphi]-\Scal_{\p\Om}\Kcal^*_{\p\Om}[\varphi]\notag\\
&=\lambda \Scal_{\p\Om}[\varphi]-\Kcal_{\p\Om}\Scal_{\p\Om}[\varphi]\\
&=\left(\lambda +\frac{1}{2}\right)\Scal_{\p\Om}[\varphi]-\Dcal_{\p\Om}\left[\Scal_{\p\Om}[\varphi]\right]\Big|^-\\
&=\left(\lambda+\frac{1}{2}\right)(u-H)-\Dcal_{\p\Om}\left[(u-H)|_{\p\Om}\right]\Big|^-\quad\mbox{on }\p\Om.
\end{align*}
Because both sides are harmonic in $\Om$, we have
\beq\label{eqn:Scal:nuH}
\Scal_{\p\Om}[\nu\cdot\nabla H]=\Big(\lambda+\frac{1}{2}\Big)(u-H)-\Dcal_{\p\Om}\left[(u-H)|_{\p\Om}\right]\quad\mbox{in } \Om.
\eeq

\noindent (b)$\,\Rightarrow\,$(a). If $(\lambda+\frac{1}{2})(u-H)-\Dcal_{\p\Om}\left[(u-H)|_{\p\Om}\right]$ is a harmonic polynomial of degree $N$ in $\Om$, then so is $\Scal_{\p\Om}[\nu\cdot \nabla H]$. From the discussion at the beginning of section~\ref{subsection:mainproof}, we have
\beq\label{nuH:span:N}
\nu\cdot\nabla H\in \mbox{span}\left(\psi_{-N},\psi_{-N+1},\cdots,\psi_{N-1},\psi_N\right)
\eeq
 {and $\left\la \nu\cdot\nabla H,\psi_{- N}\right\ra \neq 0$.}
From \eqnref{theorem:general:Eshelby:eqn1} and Lemma\,\ref{thm:series} (b), we obtain
\beq\label{thm2:Kpsi1:N}
%\Kcal^*_{\p\Om}[\psi_1]\in\mbox{span}\left(\psi_{-N},\psi_{-N+1},\cdots,\psi_{N-1},\psi_N\right)
\left\la \Kcal^*_{\p\Om}[\psi_1],\psi_{-k}\right\ra=0\quad\mbox{for all }k\geq N+1
\eeq
and $\left\la \Kcal^*_{\p\Om}[\psi_1],\psi_{- N}\right\ra \neq 0$ (except $N=1$).
Hence, $\Om$ is a domain of order $N$ from Lemma~\ref{psi_1_poly}.
\smallskip

\noindent(a)$\,\Rightarrow\,$(b). Assume that $\Om$ is a domain of order $N$. In fact, we can show that $(\lambda+\frac{1}{2})(u-H)-\Dcal_{\p\Om}\left[(u-H)|_{\p\Om}\right]$ is a harmonic polynomial of degree $N$ in $\Om$ for any uniform loading $H$. From the assumption on $\Om$, it is smooth and \eqnref{thm2:Kpsi1:N} holds. For any uniform loading $H$, from \eqnref{theorem:general:Eshelby:eqn1}, we have \eqnref{nuH:span:N}. From \eqnref{eqn:Scal:nuH}, we deduce that $(\lambda+\frac{1}{2})(u-H)-\Dcal_{\p\Om}\left[(u-H)|_{\p\Om}\right]$ is a harmonic polynomial of degree $N$ in $\Om$.

\qed

\subsection{Proof of Theorem~\ref{theorem:iso:polnomial:b}}

From \eqnref{eqn:main_relation}, we have
$$(\lambda I - \Kcal^*_{\p\Om})[{\psi_u}] =\left( \lambda- \frac{1}{2} \right)\psi_H.$$
As generalizations of \eqnref{psi:H:expan} and \eqnref{psi:u:expan}, let $\psi$ and $\widetilde{\psi}$ have series expansion of infinite order such that
$$
{\psi_H}=\sum_{m=1}^{\infty} \left(\alpha_m m \gamma^m\psi_m+\overline{\alpha_m}m\gamma^m\psi_{-m}\right), \quad
{\psi_u}=\sum_{m=1}^{\infty} \left(\beta_m m \gamma^m \psi_m+\overline{\beta_m} m \gamma^m\psi_{-m}\right).
$$
Then we have
$$(\lambda I - \Kcal^*_{\p\Om})\left[\sum_{m=1}^\infty m \gamma^m \left(\beta_m \psi_m+\overline{\beta_m} \psi_{-m}\right)\right] = \left( \lambda- \frac{1}{2} \right) \sum_{m=1}^\infty m \gamma^m \left(\alpha_m \psi_m+\overline{\alpha_m} \psi_{-m}\right).$$
By using Lemma \ref{thm:series}(b), we obtain
\begin{align*}
&\lambda \sum_{m=1}^\infty m \gamma^m \left(\beta_m \psi_m+\overline{\beta_m} \psi_{-m}\right) - \frac{1}{2} \sum_{m=1}^\infty m \gamma^{-m} \left(\sum_{k=1}^\infty \overline{\beta_k c_{km}} \psi_{m} +  \sum_{k=1}^\infty \beta_k c_{km} \psi_{-m}  \right)\\
& = \left( \lambda- \frac{1}{2} \right) \sum_{m=1}^\infty m \gamma^m \left(\alpha_m\psi_m+\overline{\alpha_m} \psi_{-m}\right).
\end{align*}
Hence, we have
$$\left( \lambda- \frac{1}{2} \right) \alpha_m = \lambda \beta_m - \frac{1}{2} \sum_{k=1}^\infty \overline{\beta_k c_{km}} \gamma^{-2m}.$$
In other words,
\beq\label{rel}
\left( \lambda- \frac{1}{2} \right) \bm{\alpha}  = \lambda \bm{\beta}   - \frac{1}{2} \overline{\bm{\beta} C} \gamma^{-2\NN},
\eeq
where $\bm{\alpha}  = (\alpha_m)_{m=1}^\infty$ and $\bm{\beta}  = (\beta_m)_{m=1}^\infty$ are row vectors, $C=(c_{km})_{k,m=1}^\infty$ is a semi-infinite matrix, and $\gamma^{-2\NN}$ is a diagonal matrix for which the $(m,m)$-entry is $\gamma^{-2m}$.

The conjugate on both sides of \eqnref{rel} is as follows,
\beq\label{rela}
 \overline{\bm{\beta} } =  \left( 1- \frac{1}{2\lambda} \right) \overline{\bm{\alpha}} + \frac{1}{2\lambda} \bm{\beta}  C \gamma^{-2\NN}.
\eeq
By substituting \eqnref{rela} into \eqnref{rel}, we finally obtain
$$\bm{\beta}  = \left( \lambda-\frac{1}{2} \right) \left( \lambda\bm{\alpha}\gamma^{2\NN} + \frac{1}{2} \overline{\bm{\alpha} C} \right)  \left( \lambda^2 I - \frac{\gamma^{-2\NN}C\gamma^{-2\NN}\overline{C}}{4} \right)^{-1} \gamma^{-2\NN}.$$
For the invertibility of $\lambda^2 I - \frac{\gamma^{-2\NN}C\gamma^{-2\NN}\overline{C}}{4}$, we refer the reader to see \cite{Choi:2020:ASR}.

\qed

\section{Conclusion}\label{ref:conclusion}
In this paper, we investigated the Eshelby uniformity principle for anti-plane elasticity. We extended the uniformity principle to domains of general shape with polynomial loadings by using the series expression of the solution to the transmission problem, which was recently developed in \cite{Jung:2018:SSM}.
Also, we derived an explicit expression for the solution in matrix form using the Grunsky coefficients for the isotropic case.

%\bibliography{2020_Eshelby_bib}{}

\begin{thebibliography}{10}

\bibitem{Ammari:2013:MSM:book}
Habib Ammari, Josselin Garnier, Wenjia Jing, Hyeonbae Kang, Mikyoung Lim, Knut
  S{\o}lna, and Han Wang.
\newblock {\em Mathematical and statistical methods for multistatic imaging}.
\newblock Lecture Notes in Mathematics (vol. 2098). Springer, Cham, 2013.

\bibitem{Ammari:2004:RSI:book}
Habib Ammari and Hyeonbae Kang.
\newblock {\em Reconstruction of small inhomogeneities from boundary
  measurements}.
\newblock Lecture Notes in Mathematics (vol. 1846). Springer-Verlag, Berlin,
  2004.

\bibitem{Ammari:2019:SRN}
Habib Ammari, Mihai Putinar, Matias Ruiz, Sanghyeon Yu, and Hai Zhang.
\newblock Shape reconstruction of nanoparticles from their associated plasmonic
  resonances.
\newblock {\em J. Math. Pures Appl.}, 122:23 -- 48, 2019.

\bibitem{Bardsley:2017:CGB}
P.~Bardsley, M.~S. Primrose, M.~Zhao, J.~Boyle, N.~Briggs, Z.~Koch, and G.~W.
  Milton.
\newblock Criteria for guaranteed breakdown in two-phase inhomogeneous bodies.
\newblock {\em Inverse Problems}, 33(8):085006, 2017.

\bibitem{Bieberbach:1916:KDP}
L.~Bieberbach.
\newblock {\"U}ber die {K}oeffizienten derjenigen {P}otenzreihen, welche eine
  schlichte {A}bbildung des {E}inheitskreises vermitteln.
\newblock {\em Sitzungsberichte Preussische Akademie der Wissenschaften},
  138:940--955, 1916.

\bibitem{Caratheodory:1913:GBR}
C.~Carath\'{e}odory.
\newblock \"{U}ber die gegenseitige {B}eziehung der {R}\"{a}nder bei der
  konformen {A}bbildung des {I}nneren einer {J}ordanschen {K}urve auf einen
  {K}reis.
\newblock {\em Math. Ann.}, 73(2):305--320, 1913.

\bibitem{Chen:2015:NEE}
Fengjuan Chen, Albert Giraud, Igor Sevostianov, and Dragan Grgic.
\newblock Numerical evaluation of the eshelby tensor for a concave
  superspherical inclusion.
\newblock {\em Int. J. Eng. Sci.}, 93:51--58, 2015.

\bibitem{Cherepanov:1974:IPP}
G.P. Cherepanov.
\newblock Inverse problems of the plane theory of elasticity: Pmm vol. 38, no
  6, 1974, pp. 963--979.
\newblock {\em J. Appl. Math. Mech.}, 38(6):915 -- 931, 1974.

\bibitem{Choi:2020:ASR}
Doosung Choi, Junbeom Kim, and Mikyoung Lim.
\newblock Analytical shape recovery of a conductivity inclusion based on
  {Faber} polynomials.
\newblock {\em To appear in Mathematische Annalen}.

\bibitem{Choi:2018:GME:preprint}
Doosung Choi, Junbeom Kim, and Mikyoung Lim.
\newblock Geometric multipole expansion and its application to neutral
  inclusions of general shape.
\newblock {\em arXiv preprint arXiv:1808.02446}, 2018.

\bibitem{Duren:1983:UF}
P.L. Duren.
\newblock {\em Univalent Functions}.
\newblock Grundlehren der mathematischen Wissenschaften (vol. 259).
  Springer-Verlag New York, 1983.

\bibitem{Escauriaza:1992:RTW}
L.~Escauriaza, E.~B. Fabes, and G.~Verchota.
\newblock On a regularity theorem for weak solutions to transmission problems
  with internal {L}ipschitz boundaries.
\newblock {\em Proc. Amer. Math. Soc.}, 115(4):1069--1076, 1992.

\bibitem{Escauriaza:1993:RPS}
Luis Escauriaza and Jin~Keun Seo.
\newblock Regularity properties of solutions to transmission problems.
\newblock {\em Trans. Amer. Math. Soc.}, 338(1):405--430, 1993.

\bibitem{Eshelby:1957:DEF}
J.~D. Eshelby.
\newblock The determination of the elastic field of an ellipsoidal inclusion,
  and related problems.
\newblock {\em Proc. Roy. Soc. London Ser. A}, 241(1226):376--396, 1957.

\bibitem{Eshelby:1961:EII}
J.~D. Eshelby.
\newblock Elastic inclusions and inhomogeneities.
\newblock In {\em Progress in Solid Mechanics}, volume~II, pages 87--140.
  North-Holland Publishing Co., Amsterdam, 1961.

\bibitem{Faber:1903:PE}
Georg Faber.
\newblock {\"U}ber polynomische {Entwickelungen}.
\newblock {\em Math. Ann.}, 57(3):389--408, Sep 1903.

\bibitem{Gao:2010:SGS}
X.-L. Gao and H.~M. Ma.
\newblock Strain gradient solution for {E}shelby's ellipsoidal inclusion
  problem.
\newblock {\em Proc. R. Soc. Lond. Ser. A}, 466(2120):2425--2446, 2010.

\bibitem{Grabovsky:1995:MMEb}
Yury Grabovsky and Robert~V. Kohn.
\newblock Microstructures minimizing the energy of a two phase elastic
  composite in two space dimensions. {II}: {The} {Vigdergauz} microstructure.
\newblock {\em J. Mech. Phys. Solids}, 43(6):949--972, June 1995.

\bibitem{Huang:2011:EEE}
Mojia Huang, Ping Wu, Guoyang Guan, and Wenchang Liu.
\newblock Explicit expressions of the {Eshelby} tensor for an arbitrary {3D}
  weakly non-spherical inclusion.
\newblock {\em Acta Mechanica}, 217(1):17--38, Feb 2011.

\bibitem{Huang:2009:EEE}
Mojia Huang, Wennan Zou, and Quan-Shui Zheng.
\newblock Explicit expression of {Eshelby tensor} for arbitrary weakly
  non-circular inclusion in two-dimensional elasticity.
\newblock {\em Int. J. Eng. Sci.}, 47(11):1240 -- 1250, 2009.

\bibitem{Jung:2018:SSM}
YoungHoon Jung and Mikyoung Lim.
\newblock A new series solution method for the transmission problem.
\newblock {\em arXiv preprint arXiv:1803.09458}, 2018.

\bibitem{Jung:2020:DEE}
YoungHoon Jung and Mikyoung Lim.
\newblock A decay estimate for the eigenvalues of the {N}eumann-{P}oincar\'{e}
  operator using the grunsky coefficients.
\newblock {\em Proc. Amer. Math. Soc.}, 148(2):591--600, 2020.

\bibitem{Kang:2009:CPS}
Hyeonbae Kang.
\newblock Conjectures of {P{\'o}lya--Szeg{\"{o}}} and {Eshelby}, and the
  {Newtonian} potential problem: A review.
\newblock {\em Mech. Mater.}, 41(4):405--410, April 2009.

\bibitem{Kang:2008:IPS}
Hyeonbae Kang, Eunjoo Kim, and Graeme~W. Milton.
\newblock Inclusion pairs satisfying {Eshelby}'s uniformity property.
\newblock {\em SIAM J. Appl. Math.}, 69(2):577--595, 2008.

\bibitem{Kang:2011:SBV}
Hyeonbae Kang, Eunjoo Kim, and Graeme~W. Milton.
\newblock Sharp bounds on the volume fractions of two materials in a
  two-dimensional body from electrical boundary measurements: the translation
  method.
\newblock {\em Calc. Var. Partial Differential Equations}, 45(3--4):367--401,
  November 2012.

\bibitem{Kang:2008:SPS}
Hyeonbae Kang and Graeme~W. Milton.
\newblock Solutions to the {P{\'o}lya--Szeg{\"{o}}} conjecture and the {weak
  Eshelby conjecture}.
\newblock {\em Arch. Ration. Mech. Anal.}, 188(1):93--116, April 2008.

\bibitem{Kellogg:1953:FPT}
Oliver~Dimon Kellogg.
\newblock {\em Foundations of potential theory}.
\newblock Dover, New York, 1953.
\newblock Originally published in 1929 by J. Springer.

\bibitem{Kenig:1994:HAT}
Carlos~E. Kenig.
\newblock {\em Harmonic analysis techniques for second order elliptic boundary
  value problems}.
\newblock CBMS Regional Conference Series in Mathematics (vol. 83). American
  Mathematical Society, Providence, RI, 1994.

\bibitem{Kwon:2002:RTA}
Ohin Kwon, Jin~Keun Seo, and Jeong-Rock Yoon.
\newblock A real-time algorithm for the location search of discontinuous
  conductivities with one measurement.
\newblock {\em Comm. Pure Appl. Math.}, 55(1):1--29, 2002.

\bibitem{Lee:1977:ESE}
Jong~K. Lee and William~C. Johnson.
\newblock Elastic strain energy and interactions of thin square plates which
  have undergone a simple shear.
\newblock {\em Scripta Metallurgica}, 11(6):477 -- 484, 1977.

\bibitem{Lee:2016:EEF}
Y.-G. Lee and W.-N. Zou.
\newblock Exterior elastic fields of non-elliptical inclusions characterized by
  laurent polynomials.
\newblock {\em Eur. J. Mech. A-Solid}, 60:112 -- 121, 2016.

\bibitem{Lim:2018:NNC}
Mikyoung Lim and Graeme~W Milton.
\newblock Non-elliptical neutral coated inclusions with anisotropic
  conductivity and {$E_\Omega$}-inclusions of general shapes.
\newblock {\em arXiv preprint arXiv:1809.10373}, 2018.

\bibitem{Liu:2008:SEC}
L.~P. Liu.
\newblock Solutions to the {Eshelby} conjectures.
\newblock {\em Proc. Roy. Soc. Ser. A}, 464(2091):573--594, May 2008.

\bibitem{Lubarda:1998:AEP}
V.~A. Lubarda and X.~Markenscoff.
\newblock On the absence of {Eshelby} property for non-ellipsoidal inclusions.
\newblock {\em Int. J. Solids. Struct.}, 35:3405--3411, 1998.

\bibitem{Markenscoff:1997:SEI}
Xanthippi Markenscoff.
\newblock On the shape of the {Eshelby} inclusions.
\newblock {\em Journal of Elasticity}, 49(2):163--166, 1997.

\bibitem{Markenscoff:1998:ICE}
Xanthippi Markenscoff.
\newblock Inclusions with constant eigenstress.
\newblock {\em J. Mech. Phys. Solids.}, 46(12):2297--2301, 1998.

\bibitem{Milton:2001:NCI}
Graeme~W. Milton and Sergey~K. Serkov.
\newblock Neutral coated inclusions in conductivity and {anti-plane}
  elasticity.
\newblock {\em Proc. Roy. Soc. London Ser. A}, 457(2012):1973--1997, 2001.

\bibitem{Onaka:2001:AET}
Susumu Onaka.
\newblock Averaged {Eshelby} tensor and elastic strain energy of a
  superspherical inclusion with uniform eigenstrains.
\newblock {\em Phil. Mag. Lett.}, 81(4):265--272, 2001.

\bibitem{Pommerenke:1992:BBC:book}
Christian Pommerenke.
\newblock {\em Boundary behaviour of conformal maps}.
\newblock Grundlehren der mathematischen Wissenschaften (vol. 299).
  Springer-Verlag, Berlin, Germany, 1992.

\bibitem{Rodin:1996:EIP}
Gregory~J. Rodin.
\newblock {Eshelby's} inclusion problem for polygons and polyhedra.
\newblock {\em J. Mech. Phys. Solids.}, 44(12):1977--1995, 1996.

\bibitem{Ru:1999:ASE}
C.~Q. Ru.
\newblock Analytic solution for {Eshelby's} problem of an inclusion of
  arbitrary shape in a plane or half-plane.
\newblock {\em J. Appl. Mech.}, 66:315--322, 1999.

\bibitem{Ru:2000:EPT}
C.~Q. Ru.
\newblock {Eshelby's} problem for two-dimensional piezoelectric inclusions of
  arbitrary shape.
\newblock {\em Proc. Roy. Soc. London Ser. A}, 456(1997):1051--1068, 2000.

\bibitem{Ru:1996:EIA}
Chong-Qing Ru and Peter Schiavone.
\newblock On the elliptic inclusion in anti-plane shear.
\newblock {\em Math. Mech. Solids}, 1(3):327--333, September 1996.

\bibitem{Sendeckyj:1970:EIP}
G.~P. Sendeckyj.
\newblock Elastic inclusion problems in plane elastostatics.
\newblock {\em Int. J. Solids. Struct.}, 6(12):1535--1543, December 1970.

\bibitem{Smirnov:1968:FCV}
V.~I. Smirnov and N.~A. Lebedev.
\newblock {\em Functions of a complex variable: {C}onstructive theory}.
\newblock Translated from the Russian by Scripta Technica Ltd. The M.I.T.
  Press, Cambridge, Mass., 1968.

\bibitem{Verchota:1984:LPR}
Gregory Verchota.
\newblock Layer potentials and regularity for the {D}irichlet problem for
  {L}aplace's equation in {L}ipschitz domains.
\newblock {\em J. Funct. Anal.}, 59(3):572--611, 1984.

\bibitem{Vigdergauz:1976:IEI}
S.B. Vigdergauz.
\newblock Integral equation of the inverse problem of the plane theory of
  elasticity.
\newblock {\em J. Appl. Math. Mech.}, 40(3):518--522, 1976.

\bibitem{Wang:2018:EIA}
Xu~Wang, Liang Chen, and Peter Schiavone.
\newblock {Eshelby} inclusion of arbitrary shape in isotropic elastic materials
  with a parabolic boundary.
\newblock {\em J. Mech. Mater. Struct.}, 13(2):191--202, 2018.

\bibitem{Wang:2019:ECE}
Xu~Wang and Peter Schiavone.
\newblock Effect of a circular {Eshelby} inclusion on the non-elliptical shape
  of a coated neutral inhomogeneity with internal uniform stresses.
\newblock {\em ZAMM - Journal of Applied Mathematics and Mechanics},
  99(6):e201900058, 2019.

\bibitem{Zou:2010:EPN}
Wennan Zou, Qichang He, Mojia Huang, and Quanshui Zheng.
\newblock {Eshelby's} problem of non-elliptical inclusions.
\newblock {\em J. Mech. Phys. Solids}, 58(3):346 -- 372, 2010.

\end{thebibliography}
%\bibliographystyle{plain}
%
\ifx \bblindex \undefined \def \bblindex #1{} \fi\ifx \bbljournal \undefined
  \def \bbljournal #1{{\em #1}\index{#1@{\em #1}}} \fi\ifx \bblnumber
  \undefined \def \bblnumber #1{{\bf #1}} \fi\ifx \bblvolume \undefined \def
  \bblvolume #1{{\bf #1}} \fi\ifx \noopsort \undefined \def \noopsort #1{} \fi

\end{document}